\documentclass[12pt]{article} 

\usepackage[a4paper,left=2cm,right=2cm,top=2.5cm,bottom=2.5cm]{geometry}
\usepackage{graphicx} 
\usepackage[multiple]{footmisc}
\usepackage{float} 
\usepackage{tcolorbox}
\usepackage{wrapfig} 
\usepackage{natbib}
\usepackage{colortbl}
\usepackage{collcell}
\usepackage{mathtools}
\usepackage{amsmath}
\usepackage{amssymb}
\usepackage{hyperref}
\usepackage{titling}
\usepackage{tikz,pgfplots}

\usepackage[spanish,english]{babel}
\usepackage{subcaption}
\usetikzlibrary{patterns}
\usetikzlibrary{shapes.geometric}
\usetikzlibrary{arrows}
\usetikzlibrary{patterns.meta}
\usepgfplotslibrary{fillbetween}
\pgfplotsset{compat=1.10}

\usepackage{bbold}
\usepackage[utf8]{inputenc}

\usepackage{amsthm}
\usepackage{amsfonts}
\usepackage{chngcntr}
\usepackage{apptools}

\AtAppendix{\counterwithin{lemma}{section}}
\usepackage[title,titletoc,toc]{appendix}

\newtheorem{theorem}{Theorem}[section]
\newtheorem{corollary}{Corollary}[theorem]
\newtheorem{lemma}[theorem]{Lemma}
\newtheorem{proposition}[theorem]{Proposition}

\newtheorem{definition}[theorem]{Definition}

\newtheorem{innercustomgeneric}{\customgenericname}
\providecommand{\customgenericname}{}
\newcommand{\newcustomtheorem}[2]{%
	\newenvironment{#1}[1]
	{%
		\renewcommand\customgenericname{#2}%
		\renewcommand\theinnercustomgeneric{##1}%
		\innercustomgeneric
	}
	{\endinnercustomgeneric}
}

\newcustomtheorem{customprop}{Proposition}

\DeclareMathOperator*{\argmax}{argmax} 

\graphicspath{{Pictures/}} 

\begin{document}
	
	\newgeometry{top=10mm, bottom=20mm}

\title{\textbf{Trade among moral agents with information asymmetries}\thanks{I am grateful to Ingela Alger, Eric Danan, Pierre Fleckinger, Ernesto Gavassa-Pérez, Peter Hammond, Pau Juan-Bartroli, Frederick Koessler, Pauline Morault, Esteban Mu\~{n}oz-Sobrado, R\'{e}gis Renault and Boris van Leeuwen for helpful discussions. I also thank audiences at the Economic Science Association (ESA Bologna 2022) and the Foundations of Utility and Risk (FUR Brisbane 2024) conferences, at seminars at Catholic University of Uruguay, CY Cergy Paris University, Universidad de la Rep\'{u}blica, and Uruguayan Central Bank for feedback.}}
\date{ \normalsize \today }

\author{Jos\'e Ignacio Rivero-Wildemauwe\thanks{Department of Management and Business, Universidad Católica del Uruguay, Montevideo, Uruguay. \\\color{blue}joseignacio.rivero@ucu.edu.uy}}

\maketitle
\thispagestyle{empty}
{\small \noindent {\textbf{Abstract:} Two agents trade an item in a simultaneous offer setting, where the exchange takes place if and only if the buyer's bid price weakly exceeds the seller's ask price. Each agent is randomly assigned the buyer or seller role. Both agents are characterized by a certain degree of Kantian morality, whereby they pick their bidding strategy behind a Veil of Ignorance, taking into account how the outcome would be affected if their trading partner adopted their strategy. I consider two variants with asymmetric information, respectively allowing buyers to have private information about their valuation or sellers to be privately informed about the item's quality. I show that when all trades are socially desirable, even the slightest degree of morality guarantees that the outcome is fully efficient. In turn, when quality is uncertain and some exchanges are socially undesirable, full efficiency is only achieved with sufficiently high moral standards. Moral concerns also ensure equal \textit{ex-ante} treatment of the two agents in equilibrium. Finally, I show that if agents are altruistic rather than moral, inefficiencies persist even with a substantial degree of altruism.}}

\medskip 

{\noindent \small \textbf{JEL codes}: D03; D82; D91; C78
	
\noindent \textbf{Keywords:} bilateral trade; adverse selection; \textit{homo moralis}; Veil of Ignorance; altruism }
\restoregeometry

\newpage

	\section{Introduction}\label{intro} 
	
	Asymmetric information constitutes one of the main departures from perfect competition. It can produce significant efficiency losses in trade, as mutually beneficial exchanges may fail to occur. The \cite{myerson1983} theorem is a stark illustration of this issue. While the effects of information asymmetries have been largely studied through models populated by purely self-interested materialistic agents, there is now substantial evidence documenting deviations from pure material self-interest and a vast literature modelling this behaviour (see reviews by \cite{fehr2003} or  \cite{camerer2003}). 
	
This paper focuses more specifically on the implications of agents weighing how they themselves would fare if \textit{hypothetically}, other individuals also adopted their own course of action. This kind of universalization reasoning is potentially very consequential in bilateral trade interactions, as it might prevent agents from fully exploiting informational advantages, rendering the outcomes of such interactions more efficient. The existing theoretical literature already shows that universalization ethics concerns help to remedy market failures in a host of different settings (see \cite{alger2016b}). Moreover, there exists by now a large body of empirical evidence that justifies its serious study .\footnote{Experimental evidence consistent with the presence of Kantian morality is reported in \cite{levine2020} and Alger and Rivero-Wildemauwe (2023). Kantian preferences’ large out-of-sample predictive power is documented in \cite{miettinen2020} and \cite{vanleeuwen2023}. For evidence of deontological thinking, see for example \cite{CapraroRand2018}, \cite{CapraroTappin2018}, \cite{bursztyn2019}, \cite{Capraro2019}, \cite{bilancini2020} and \cite{sutter2020}. For the theory grounding these particular moral motivations in evolutionary processes, refer to \cite{alger2013} and \cite{alger2020}.} 

	Here I analyse the extent to which universalization ethics concerns (or \textit{homo moralis} preferences) \textit{à la} \cite{alger2013} can solve inefficiencies originated in an unequal distribution of information in bilateral trade.\footnote{Theirs is not the only way to consider universalization ethics in such a setting. A very relevant alternative is the approach described in \cite{roemer2010}, where the notion of a ``Kantian Equilibrium'' is proposed. This is defined as a strategy profile where no player would like all players to deviate in the same way.} In this context, an agent who finds herself in the \textit{Seller} role will take into account her own payoff if she were a buyer and happened to meet a seller who employs her own strategy. Likewise, an individual in the \textit{Buyer} role applies the same reasoning with the roles reversed. Exactly how much agents care about this counter-factual scenario is measured by a parameter henceforth referred to as the \textit{degree of morality}. 
	 
	In order to capture this precise setting, I assume that agents play a bilateral trading game where one is cast into the \textit{Buyer} role and the other one into the \textit{Seller} role, but they make their decisions behind a \textit{Veil of Ignorance}. That is, without knowing what the final role distribution will be. This in turns implies that each participant needs to choose their \textit{Seller} strategy as well as their \textit{Buyer} strategy before learning their role.

	After roles have been assigned, agents play a game where they trade an item in a simultaneous offer setting. The exchange takes place if and only if the \textit{Buyer}’s bid price weakly exceeds the \textit{Seller}’s ask price. This simple setting manages to effectively capture bargaining processes where agents are able to commit to not accepting offers worse than the ones they stipulate in their strategy prior to entering the negotiation process proper.  
	
	Due to the simultaneous-move nature of the game, the framework does not assume asymmetries in bargaining power (\cite{wildemauwe2023} studies a sequential game between potentially moral where the \textit{Buyer} holds full bargaining power). It also incorporates information asymmetries between \textit{Buyer} and \textit{Seller}. More precisely, I consider two kinds of canonical information asymmetries that have been well studied under the assumption that agents only strive to maximise their own material welfare. The first one is the case where the \textit{Buyer}'s valuation for the good is her own private information and the \textit{Seller} can only resort to uniform pricing. The second one is adverse selection. In this scenario, the information asymmetry favours the \textit{Seller}, who knows the quality of her product (which may be high or low), while the \textit{Buyer} can only form an expectation about it. 

	In the case where the \textit{Seller} is uncertain about the \textit{Buyer}'s valuation, the model always produces two kinds of equilibria. The first one is fully efficient as no consumers are excluded from the market. This result is a consequence of the one-shot nature of the game I consider. The second one features low-valuation consumers being excluded from the market (as in the standard monopoly model) even though costs are below the lowest consumer valuation. Introducing moral concerns completely mutes this source of inefficiency, with only the most efficient profiles (those where all consumers are served) surviving as equilibria.

	In the adverse selection setting I consider two variants. Firstly, I examine the case where the trade of both qualities produces a net positive surplus. Secondly I tackle the case where trade of the low quality item is socially undesirable, as its cost is higher than the consumer's valuation for it.
	
	When trade of both qualities produces a net positive surplus, a sufficiently high expected quality leads to equilibria that feature either no trade, only trade of the low quality or trade of both qualities (the latter being the efficient result). In turn, if expected quality is low, only equilibria with no trade or where high quality sellers are driven out of the market are possible. Introducing even the slightest degree of morality (as long as agents are not completely Kantian) fully restores efficiency, as only equilibria with trade of both qualities are possible in this case. 
	
	So far, both the case where the \textit{Seller} is uncertain about the \textit{Buyer}'s valuation and the adverse selection setting with desirable trade lead to the same conclusion: going from a game populated by agents only concerned with their own material payoff (\textit{homo oeconomicus}) to the same one but with agents holding moral concerns eliminates all but the efficient equilibria. These surviving equilibria are also symmetric in strategies, producing an \textit{ex-ante} egalitarian result. It can be stated then that introducing \textit{homo moralis} preferences in these contexts (starting from the benchmark self-interested materialistic agents) narrows down the equilibrium set, rather than introducing equilibrium behaviour that would not arise with \textit{homo oeconomicus}. It is worth underscoring that for this narrowing down to occur, all it takes is for agents to be even slightly morally concerned. In other words, the model does not feature a positive threshold degree of morality for the efficient equilibria to be the only outcome possible, as I instead find a stark discontinuity at zero, which implies that even the smallest degree of morality suffices to eliminate all inefficient equilibria.
	
	The situation where trade of the low quality item is undesirable presents important differences with respect to the previous ones. Firstly, the most efficient result, which is having only trade of the high quality item, is not an equilibrium when agents are \textit{homo oeconomicus}. However, a sufficiently high degree of morality leads to these kind of profiles becoming equilibria and moreover, the only kind possible. Thus, in this case morality does not act as a mere equilibrium selection device. Rather, it allows for the existence of a fully efficient equilibrium and, given a high enough degree of morality, it also selects on it.
	
	Finally, I return to all the model variants described above and compare my results with the case where players are endowed with altruistic preferences \textit{à la} \cite{becker1976}. The general conclusion is that it takes a relatively high degree of altruism to fully restore efficiency even in the settings where all trade is desirable. This result contrasts sharply with the effects of morality, as in most cases partially Kantian agents manage to coordinate on the efficient equilibria regardless of their degree of morality.

	This paper is broadly related to the theoretical literature that studies the effects of pro-social preferences and moral motivations on the equilibrium outcomes in a host of strategic interactions (see e.g. \cite{arrow73}, \cite{becker1976}, \cite{andreoni1990}, \cite{bernheim1994}, \cite{levine1998}, \cite{fehr1999}, \cite{akerlof2000}, \cite{bernabeu2006}, \cite{alger2007}, \cite{ellingsen2008}, \cite{englmaier2010}, \cite{dufwenberg2011}). However, the closest references are those that analyse the effects of moral concerns \textit{à la} \cite{alger2013}. 		

	In this line, \cite{alger2017} study several canonical games and compare the outcomes when players have standard, altruistic and moral preferences. Meanwhile, \cite{sarkisian2017} analyses the optimal contract in a moral hazard problem where a risk-neutral principal hires two risk-averse agents to perform a joint task and may be endowed with either standard, altruistic and moral preferences. \cite{ayoubi2020a} rely on Kantian moral concerns to explain why agents may engage in costly pro-environmental behaviour. \cite{munoz2022} tackles optimal linear and non-linear taxation problems with moral agents. Finally, \cite{juan2024moral} study a standard Nash bargaining game where individuals have moral concerns.
	
	The aforementioned references have in common the fact that they deal with symmetric games where information is equally distributed. In contrast, this paper focuses on contexts with different information structures and importantly, where information asymmetries are central. This in turn implies that the very nature of the interaction is by definition asymmetric. The \textit{Veil of Ignorance} approach that I formalise here allows me to carry out the analysis of this \textit{a priory} asymmetric situation through the lens of \textit{homo moralis}, and thus constitutes a contribution to this literature. This modelling strategy thus allows me to document the first set of results on how moral concerns interact with asymmetric information in a bilateral trade setting.

	The rest of the paper is organised as follows. Section \ref{sec:prelim} presents the general framework and the benchmark bilateral trade game under complete information. Section \ref{sec:incompinfo:valunc} extends the model to consider the case where buyers are heterogeneous in their valuations but the \textit{Seller} can only set one price. Section \ref{sec:incompinfo:quality} analyses the adverse selection case, where the buyer does not know the product's quality. Section \ref{sec:alt} revisits the previous models now populated by altruistic agents. Finally, Section \ref{sec:disc} provides an overall discussion of my results.

%
%
\section{Framework} \label{sec:prelim}

		As stated in the Introduction, a simple and tractable way to consider agents capable of universalising their choices in an asymmetric setting is to take an \textit{ex-ante} approach, analysing decisions made behind a  ``Veil of Ignorance'' with respect to the role distribution. In this line, I set up an environment where two identical agents play one of two games, with Nature deciding which one is to be played: either the one that has Player 1 in the first role (say, \textit{Seller}) and therefore Player 2 in the second role (\textit{Buyer}) or vice-versa.
				
		More formally, consider $G^1$, a ``contingent game'' between players $1$ and $2$. Player 1 chooses strategy $s_1 \in S$ and Player 2 chooses strategy $b_2 \in B$. They respectively obtain the payoffs $\pi^s(s_1,b_2)$ and  $\pi^b(s_1,b_2)$, where $\pi^s,\pi^b: S\times B \rightarrow \mathbb{R}$. In addition, take the contingent game $G^2$, a game between players $1$ and $2$ where Player 1 chooses strategy $b_1 \in B$ and Player 2 chooses strategy $s_2 \in S$. They respectively obtain the payoffs $\pi^b(s_2,b_1)$ and  $\pi^s(s_2,b_1)$.
		
		Next, define $G$ as the game where Nature randomly chooses whether $G^1$ or $G^2$ is played, with equal probabilities. The \textit{ex-ante} material payoff obtained by Player $i\in\{1,2\}$ when using strategy $(s_i,b_i)\in S\times B$ against Player $j$'s strategy $(s_j,b_j)\in S\times B$ (with $i\neq j$) is:
		
		\begin{equation}\label{exantepayoff}
	\pi\left((s_i,b_i),(s_j,b_j))\right) = \dfrac{\pi^s(s_i,b_j)+\pi^b(s_j,b_i)}{2}.
		\end{equation}
		
		Finally, define the \textit{utility} that Player $i$ derives from choosing  $(s_i,b_i)\in S\times B$ against $(s_j,b_j)\in S\times B$ as:		
		\begin{equation}\label{exanteutility}
			\begin{aligned}
		U\left((s_i,b_i),(s_j,b_j))\right) = (1-\kappa)\cdot\pi\left((s_i,b_i),(s_j,b_j))\right)+\kappa\cdot\pi\left((s_i,b_i),(s_i,b_i))\right),
				\end{aligned}
		\end{equation}

		\noindent where $\kappa\in[0,1]$ is a parameter representing the player's \textit{degree of morality}. When choosing a strategy in the game $G$, each agent maximizes a convex combination of her own payoff and the one she would obtain if the other player chose the same strategy as her. The term $\pi\left((s_i,b_i),(s_i,b_i)\right)$ reflects the moral element in the agent's decision-making process, as it provides an answer to the question ``what would my payoff be if the other player behaved as I do?'', while $\kappa \in [0,1]$ weighs how much the agent cares about this moral component. The second term in (\ref{exanteutility}) can thus be interpreted as a way to reflect Kant's categorical imperative, ``to act only on the maxim that you would at the same time will to be a universal law'' (\cite{kant2011}), mediated by $\kappa$.
		
		At one end of the morality spectrum are the agents with $\kappa=0$, to whom I refer as \textit{homo oeconomicus}. This is the ``standard'' case of individuals who are merely trying to maximise their own material payoff. On the other end are players with $\kappa=1$, whom (as \cite{alger2013}) I denominate \textit{homo kantiensis}. These purely Kantian agents only care about their (hypothetical) payoff when the opposing party employs the same strategy as theirs. That is, they only think about what would happen if both players in the game behaved in the same way as they do and decide irrespectively of the actual strategy chosen by the other player in their game. Finally, $\kappa \in (0,1)$ gives way to partially moral agents or \textit{homo moralis}, who put some weight on the morality of their actions while at the same time considering their own material payoff.
					
		The setting I am proposing can be used for the Veil-of-Ignorance analysis of any two-player game and is not limited to bilateral trade or bargaining processes. In order to specifically tackle a given two-player interaction, it suffices to appropriately define $S, B,\pi^s,\pi^b$ and specify $\kappa$. In this paper, contingent game $G^i$ is a game where a \textit{Seller} is in possession of an object and posts a price for it, while the \textit{Buyer} simultaneously decides which prices she is willing to accept. In this line, it can also be viewed as a \textit{contingent market} where Player $i$ is the \textit{Seller} and Player $j$ is the \textit{Buyer}. If the \textit{Seller}'s price belongs to the \textit{Buyer}'s set of acceptable prices, trade takes place at that price. If it does not, trade does not occur. 
		
%
				
	\subsection{Benchmark: Complete information} \label{sec:compinfo}

	I now present the complete information game that I will use as a benchmark. The agent assigned to the \textit{Seller} role (say, Player $i$) sets a price $p_i \in \mathbb{R}_+$. Simultaneously, the player cast as a \textit{Buyer} (say, Player $j$) chooses a threshold $\bar p_j \in \mathbb{R}_+$
	that defines her set of acceptable prices as all those that belong to the interval $[0, \bar p_j]$.\footnote{Throughout the paper, I restrict attention to strategy spaces where buyers' sets of acceptable prices are defined by thresholds.} The agent assigned to the seller role faces a cost of $r\in\mathbb{R}_+$ if they sell the object, while the player in the buyer role derives some consumption utility $v\in \mathbb{R}_+$	from acquiring the item. If no trade happens, the payoff obtained by each agent is nil. 
	
	In terms of the framework laid out at the beginning of the Section, the above describes the contingent game $G^i$, where $s_i$ is $p_i \in \mathbb{R}_+$ and $b_j$ is $\bar p_j \in \mathbb{R}_+$. This  means that $S\times B = \mathbb{R}^2_+$. The contingent payoffs $\pi^s(p_i,\bar p_j$) and $\pi^b(p_i,\bar p_j)$ are then defined as: 
	
	\begin{equation*}\label{contpayoff}
	\begin{aligned}
		\pi^s(p_i,\bar p_j)=: \mathbb{1}\{ p_i \leq \bar p_j \} (p_i - r),\\
		\pi^b(p_i,\bar p_j)=:\mathbb{1}\{p_i \leq \bar p_j \} (v - p_i) .
	\end{aligned}
	\end{equation*}
	
\noindent Substituting $\pi^s(p_i,\bar p_j$) and $\pi^b(p_i,\bar p_j)$ in expressions (\ref{exantepayoff}) and (\ref{exanteutility}) we obtain the \textit{ex-ante} payoff and utility, respectively:
	\begin{equation}\label{eq:payoff:compinfo}
		\pi\left((p_i,\overline{p}_i);(p_j,\overline{p}_j)\right)= \dfrac{1}{2}\left[\mathbb{1}\{ p_i \leq \overline{p}_j \} (p_i - r)  + \mathbb{1}\{p_j \leq \bar p_i \} (v - p_j)\right],	
	\end{equation}
	
	\begin{equation}\label{eq:compinfo:hm:ut2}
 	\begin{aligned}
	 U\left((p_i,\overline{p}_i);(p_j,\overline{p}_j)\right)=&(1-\kappa) \cdot \dfrac{1}{2}\left[\mathbb{1}\{ p_i \leq \overline{p}_j \} (p_i - r)  + \mathbb{1}\{p_j \leq \bar p_i \} (v - p_j) \right] +\\
	 &\kappa\cdot\dfrac{1}{2}\cdot\mathbb{1}\{ p_i  \leq \bar p_i \}(v-r)	 
	 	\end{aligned}
	\end{equation}

	Notice that the term inside the parenthesis in the last line of Expression (\ref{eq:compinfo:hm:ut2}) is the total net trade surplus. I consider a materially efficient outcome (or, simply put, an efficient outcome) as the one that maximises this surplus. As a result, the efficient equilibria in the complete information game are those that induce trade in both contingent markets.
	
	My main concern throughout the paper is efficiency (although I do study the equity properties of the equilibrium outcomes whenever relevant). Notice that when considering a given strategy profile of the bilateral trade game, in each contingent market it can either generate trade (if the price in that market is no larger than the threshold) or not. As a consequence, we can classify each profile into a type $A/B$ where $A$ indicates whether there is trade in the first contingent market and $B$ in the second one. For the complete information game studied in this Section, profiles can be of four types: Trade/Trade, Trade/No Trade, No Trade/Trade or No Trade/No Trade.
	
	I begin with the version of the model where $\kappa = 0$ (so players only care about their own material payoff), which I use as a benchmark. In order to find the game's Nash equilibria in pure strategies for this case, I start by noticing that, when deciding on her price if in the \textit{Seller} role, the only way this action affects an agent's payoff is through its relation with the other player's threshold. Likewise, her action as \textit{Buyer} only influences her payoff through its interaction with the other player's price. It then follows that each contingent game can be analysed separately, and the equilibrium set of the bilateral trade game is comprised of all the ordered pairs formed by an equilibrium in $G^1$ and an equilibrium in $G^2$. Proposition \ref{proposition:compinfo:selfish:maineq} summarises them.\footnote{All proofs are relegated to the Appendix.}

		\begin{proposition}\label{proposition:compinfo:selfish:maineq} 
			Consider the complete information bilateral trade game with $\kappa = 0$. The Nash equilibria in pure strategies are:		
			\begin{itemize}
				\item Trade/Trade: $\overline{p}^*_i= p_j^*\in[r,v]$ for all $i\neq j \in \{1,2\} $
				\item Trade/No Trade or No Trade/Trade: $\overline{p}^*_i= p_j^*\in[r,v] , \overline{p}^*_j<r, p^*_i > v$ for all $i\neq j \in \{1,2\} $
				\item No Trade/No Trade: $\overline{p}^*_i<r,p^*_j > v$ for all $i\neq j \in \{1,2\} $
			\end{itemize}				
		
		\end{proposition}
	
		The complete information game between two \textit{homo oeconomicus} has three types of equilibrium profiles. The first one is the Trade/Trade result, where exchanges take place in both contingent markets at prices between seller cost and consumer valuation. The second kind are No Trade/No Trade equilibria, where there is no trade in either contingent market. Here, thresholds are set below seller costs and prices are above consumer valuation. The third sort of equilibria are those with contingent trade, where there is trade in one of the contingent markets (at a price between cost and valuation) but not in the other (where the threshold is set below cost and the price is above consumer valuation).

		Recall that I assume throughout that $r<v$, so a materially efficient outcome in this model is for trade to actually happen. We see then that the second and third types of equilibrium profiles (those with no trade or with contingent trade) constitute a sort of coordination failure, a negotiation breakdown that is brought about by the simultaneity in the agents' decisions. In this regard, we could say that the particular set-up of the game is ``stacked'' against the occurrence of socially beneficial exchanges. The following sections focus on additional inefficiencies stemming from asymmetric information and explore whether they are mitigated in the presence of moral concerns.
		
		Having established this benchmark, I explore whether considering partially moral agents (for whom $\kappa \in(0,1)$) modifies this result.\footnote{The approach of solving the contingent game to find the game's equilibrium set is not applicable in this case. This is due to the fact that when setting her price, each agent needs to take not only her opponent's threshold into consideration, but also her own. As a consequence, I take the alternative and admittedly more cumbersome route of analysing each equilibrium type (no-trade, full trade and contingent trade), and discarding those profiles for which profitable deviations are available (see Lemmas \ref{lemma:compinfo:hm:contnotrade} to \ref{lemma:compinfo:hm:fulltrade1} in the Appendix). Following this, I show that the remaining candidate equilibria are indeed mutual best responses and characterise them (see Lemma \ref{lemma:compinfo:hm:fulltrade3}).} The following Proposition lays out the bilateral trade game's equilibrium set with \textit{homo moralis} agents.
	\begin{proposition}\label{prop:compinfo:hm:eq}
		Consider the complete information bilateral trade game with $\kappa \in (0,1)$. The Nash equilibria in pure strategies are $r\leq p^*_i= \bar p^*_j = p^*_j = \bar p^*_i\leq v, \forall i,j\in\{1,2\}, i\neq j$. Moreover, it suffices for only one of the players to have $\kappa>0$ to eliminate No Trade/No trade, Trade/No Trade or No Trade/Trade equilibria.
	\end{proposition}
		
	Proposition \ref{prop:compinfo:hm:eq} states that only symmetric, full-trade profiles with prices in the $[r,v]$ interval are Nash equilibria in pure strategies of the bilateral trade game. This equilibrium set is a subset of the equilibria resulting from the game between \textit{homo oeconomicus} agents.
		
	Partial morality fully restores efficiency in the complete information set-up, as it effectively precludes any no-trade or contingent-trade outcome from taking place. Even more strikingly, in order for the conclusion to go through, it is sufficient for only one of the players to be endowed with \textit{homo moralis} preferences. In fact, the attainment of full efficiency does not depend on the intensity of players' moral concerns as long as they are not fully Kantian.	
	
	It is worth further reflecting on the implications of Proposition \ref{prop:compinfo:hm:eq} by contrasting the equilibria to the version of the model with \textit{homo oeconomicus}. Recall that those profiles that constitute equilibria for the latter but not for \textit{homo moralis} are those that feature no trade in at least one contingent market, and those with trade in both but at different prices. The intuition is fairly simple. Notice that these profiles necessarily imply that at least one of the players is setting a threshold higher than her own price. While this is completely irrelevant in the \textit{homo oeconomicus} case, \textit{homo moralis} agents will find deviations from these strategies morally profitable, as they care about setting a price that they themselves would be willing to accept. Moreover, they can do this either without facing material losses or with only infinitesimally large ones (thus trading off a strictly positive moral gain against an arbitrarily small material loss in this latter case).
	
	A final interesting remark is that due to the symmetry of the bilateral trade game under analysis, a symmetric profile will result in the same \textit{ex-ante} utility. So in this way, partial moral concerns lead to complete \textit{ex-ante} equality: before the role distribution is revealed, each player can expect to obtain the same level of utility.
	
	It is reasonable to wonder whether moral motivations produce the same effects on bilateral trade's efficiency when inefficiencies result from information asymmetries between the \textit{Buyer} and the \textit{Seller}. This is a pervasive issue that ranks amongst the most studied explanations for why markets may fail to deliver fully efficient results. In particular, the case where the seller is not completely certain of the buyer's willingness to pay for the item is formalised as early as in \cite{pigou1920}'s classic analysis of monopoly with uniform pricing. In turn, the situation where buyers do not know the quality of the object being bargained over, which could either be low or high was brought forth notably by \cite{akerlof1970}. The following sections extend the complete information model described here to embed those two cases in the bargaining framework.

	\section{Asymmetric information: Heterogeneous buyer valuations} \label{sec:incompinfo:valunc}

In order to analyse the valuation uncertainty case, I re-define the contingent game $G^i$, introducing a move by Nature where it randomly decides the \textit{Buyer}'s valuation. This can be high ($v_h \in \mathbb{R}_+ $) with probability $\lambda$ or low ($v_l\in \mathbb{R}_+$) with probability $1-\lambda$. The \textit{Seller}'s cost is $r\in \mathbb{R}_+$, as in Section \ref{sec:compinfo}, and I assume that $0<r<v_l<v_h$ (so the efficient outcome is for both types of consumers to be served). The agent in the \textit{Buyer} role is informed of her valuation, while the \textit{Seller} only knows the probability $\lambda$. As a consequence, the former now sets two thresholds, while the latter sets only one price (as in Section \ref{sec:compinfo}). Contingent game $G^i$ thus represents the standard monopoly problem with uniform pricing, with the difference that price and thresholds are decided simultaneously.
	
Strategy $s_i$ is defined as $p_i \in \mathbb{R}_+$ (as in Section \ref{sec:compinfo}), while $b_j$ is defined as $(\bar p_{jh},\bar p_{jl}) \in \mathbb{R}^2_+$. Meanwhile, the contingent game's payoffs are re-defined as:
\begin{equation}\label{contpayoffvalunc}
\begin{aligned}
&\pi^s(p_i,\bar p_{jh},\bar p_{jl})=: \lambda \mathbb{1}\{ p_{i}\leq \overline{p}_{jh} \} (p_{i} - r)+(1-\lambda)\mathbb{1}\{ p_{i}\leq \overline{p}_{jl} \} (p_{i} - r)\\
&\pi^b(p_i,\bar p_{jh},\bar p_{jl})=: \lambda \mathbb{1}\{ p_{i}\leq \overline{p}_{jh} \} (v_h - p_{i}) +  (1-\lambda)\mathbb{1}\{ p_{i}\leq \overline{p}_{jl} \} (v_l - p_{i}).
\end{aligned}
\end{equation}

\noindent Substituting the contingent payoffs in Expression (\ref{exantepayoff}), we obtain the \textit{ex-ante} utility:
\begin{equation}\label{eq:compinfo:hm:ut3}
\small
\begin{aligned}
\begin{split}
U\left((p_{i},\overline{p}_{ih},\overline{p}_{il});(p_{j},\overline{p}_{jh},\overline{p}_{jl})\right)= 
& (1-\kappa) \cdot \dfrac{1}{2}\big[ \lambda \mathbb{1}\{ p_{i}\leq \overline{p}_{jh}\}(p_{i} - r) +(1-\lambda)\mathbb{1}\{ p_{i}\leq \overline{p}_{jl}\} (p_{i} - r) +\\ 
&\lambda \mathbb{1}\{ p_{j}\leq \overline{p}_{ih} \} (v_h - p_{j}) +  (1-\lambda)\mathbb{1}\{ p_{j}\leq \overline{p}_{il} \} (v_l - p_{j}) \big] +\\
&\kappa\cdot\dfrac{1}{2}\big[ \lambda \mathbb{1}\{ p_{i}\leq \overline{p}_{ih}\}(v_h - r) +(1-\lambda)\mathbb{1}\{ p_{i}\leq \overline{p}_{il}\} (v_l - r) \big]
\end{split}
\end{aligned}
\end{equation}	
	
The first two lines of Expression (\ref{eq:compinfo:hm:ut3}) state that the \textit{Seller}'s payoff in the contingent game is the weighted average of the surplus obtained by trading with either type of \textit{Buyer}, which is either the difference between price and cost or zero (if that particular consumer type is not served). The weights are each consumer type's probability. In turn, the \textit{Buyer}'s payoff is the weighted average of the surplus that she would get if she was a high valuation buyer and that corresponding to being a low valuation one. Meanwhile, the third line (which represents the moral term) is just half the expected total surplus obtained in game $G$.
	
Having two types of consumers (high and low-valuation) requires more cumbersome expressions to characterise equilibrium prices and thresholds. As a consequence, in Proposition \ref{proposition:incompinfo:valunc:selfish} I only state whether trade takes place and with which type of consumer, relegating the full characterisation to Proposition \ref{prop32star} in the Appendix. Definition \ref{def:BH} describes every kind of outcome that can arise in either contingent market. I then use this classification to state the game's equilibrium types.

	\begin{definition}\label{def:BH}
		Consider the contingent market where $i$ is in the \textit{Seller} role and $j\neq i$ is in the \textit{Buyer} role. The contingent game's outcomes can be of the following types:
		\begin{itemize}
			\item Full trade: both types of buyers are served. This outcome realizes if and only if $ p_{i}\leq \min\{\bar p_{jh},\bar p_{jl}\}$.
			\item High valuation: only high valuation buyers are served. This outcome realizes if and only if $ \bar p_{jl}<  p_{i} \leq  \bar p_{jh}$.
			\item Low valuation: only low valuation buyers are served. This outcome realizes if and only if $ \bar p_{jh}<  p_{i} \leq  \bar p_{jl}$.
			\item No trade: no buyers are served. This outcome realizes if and only if $  p_{i} > \max  \{ \bar p_{jh},\bar p_{jl}\}$.
		\end{itemize}
	\end{definition}

	\begin{proposition}\label{proposition:incompinfo:valunc:selfish}
		Consider the heterogeneous buyer valuations bilateral trade game with $\kappa = 0$. The Nash equilibria in pure strategies are of the type A/B. $A,B \in \{ \text{Full trade},\text{High val.},\allowbreak\text{No trade}\}$.	
	\end{proposition}

For a given contingent market, two of the three possible outcomes are familiar results in monopoly settings with uniform pricing. The first one is when all consumer types are being served (``Full trade''), while the second one has low-valuation consumers excluded from the market (``High valuation''). Recall that the latter is inefficient, as even the low consumer valuation is above seller cost. Unsurprisingly, the range of prices that support equilibria with exclusion is increasing in the probability of the consumer being of high valuation $\lambda$. The third kind is the negotiation failure result (``No trade''), a consequence of the simultaneity of decisions assumed in my framework that we had already encountered in the previous section. The bilateral trade game's equilibria are all the pairs that can be formed from this set. 

Introducing buyer heterogeneity in this otherwise unmodified context then naturally brings about a new source of inefficiency: asymmetric information. More specifically, the fact that the \textit{Seller} is unsure about the consumer's willingness to pay for the item makes it possible for exclusion equilibria to exist. They occur when the maximum prices accepted by the low and high valuation consumers are different enough so that, given the share of valuation consumers, the seller is better off excluding.
	
Considering partially moral agents drastically changes the types of equilibria attainable. I state it in the following Proposition, relegating the full characterisation of prices and thresholds to the Appendix, as before.\footnote{For the sake of completeness, I also present in the Appendix, in Lemma \ref{proposition:valunc:hk:maineq} the equilibria with fully moral agents. The takeaway is the same as in the complete information case.}   
	 	
	\begin{proposition}\label{proposition:incompinfo:valunc:hm}
		Consider the heterogeneous buyer valuations bilateral trade game with $\kappa \in(0,1)$. The only kind of Nash equilibria in pure strategies is Full trade/Full trade. 
		\end{proposition}			

As in Section \ref{sec:compinfo}, these profiles are a subset of the equilibria in the bilateral trading game with \textit{homo oeconomicus} agents. Similarly to the case with complete information, partial moral concerns eliminate inefficient equilibria originated in the coordination failure brought about by the one-shot assumption. But the extension analysed here shows that this is not the only effect they have. In fact, they also do away with inefficiencies stemming from the information asymmetries, as they also preclude exclusion. Importantly, the precise degree of morality is still not relevant as long as it remains strictly between 0 and 1. Therefore, we again find a marked discontinuity in the effects of introducing morality in the agents' decision making process. Going from $\kappa = 0$ to a strictly positive degree of morality, however small (and as long as it does not equal 1) suffices to do away with every inefficient equilibrium in this case as well.

	\section{Asymmetric information: Adverse selection} \label{sec:incompinfo:quality}
	
	In contrast to the previous Section, I study here a common values setting where the information asymmetry concerns both costs and valuation. More precisely, I consider a bilateral trading game where the \textit{Seller} may be in possession of either a high quality object or a low quality one. This is echoed in the consumer's valuation, which is $v_h$ for the high quality and $v_l$ for the low one. Importantly, I also assume that high-quality sellers face a higher cost than low-quality ones, namely $r_h$ and $r_l$, with $r_h>r_l$.  
	
	While the \textit{Seller} knows her item's quality, the \textit{Buyer} only knows the probability distribution. To reflect this in the definition of contingent game $G^i$, I assume that Nature first randomly chooses whether the object is of high or low quality with probabilities $\lambda$ and $1-\lambda$, respectively. The \textit{Seller} is then informed of the result of this draw, while the buyer only knows $\lambda$. In order to capture this new information structure in a tractable way, we must then have the \textit{Seller} posting two prices, one for when she is of the low quality type and one for when she is of the high quality kind. In turn, the agent in the \textit{Buyer} role posts a threshold above which she will not trade (just as in Section \ref{sec:compinfo}). All of this means that strategy $s_i$ is now defined as $(p_{ih},p_{il}) \in \mathbb{R}^2_+$, and $b_j$ as $(\bar p_j) \in \mathbb{R}_+$. The contingent game's payoffs take the form:	
	\begin{equation*}\label{contpayoffquality}
	\begin{aligned}
	&\pi^s(p_{ih},p_{il},\overline{p}_j)=: \lambda \mathbb{1}\{ p_{ih}\leq \bar p_j \} (p_{ih} - r_h) +(1-\lambda)\mathbb{1}\{ p_{il}\leq \bar p_j \} (p_{il} - r_l)\\
	&\pi^b(p_{ih},p_{il},\overline{p}_j)=: \lambda \mathbb{1}\{ p_{ih}\leq \bar p_j \} (v_h - p_{ih}) +  (1-\lambda)\mathbb{1}\{ p_{il}\leq \bar p_j \} (v_l - p_{il}).
	\end{aligned}
	\end{equation*}
		
\noindent Substituting these contingent payoffs in (\ref{exanteutility}) gives the utility function for the quality uncertainty version of the bilateral trading game:
\begin{equation}\label{eq:compinfo:hm:ut4}
	\footnotesize
	\begin{aligned}
		U\left((p_{ih},p_{il},\overline{p}_i);(p_{jh},p_{jl},\overline{p}_j)\right)= 
		& (1-\kappa) \cdot \dfrac{1}{2}\big[ \lambda \mathbb{1}\{ p_{i}\leq \overline{p}_{jh}\}(p_{i} - r_h) +(1-\lambda)\mathbb{1}\{ p_{i}\leq \overline{p}_{jl}\} (p_{i} - r_l) + \\ 
		&\lambda \mathbb{1}\{ p_{j}\leq \overline{p}_{ih} \} (v_h - p_{j}) +  (1-\lambda)\mathbb{1}\{ p_{j}\leq \overline{p}_{il} \} (v_l - p_{j}) \big] +\\
		&\kappa\cdot\dfrac{1}{2}\big[ \lambda \mathbb{1}\{ p_{i}\leq \overline{p}_{ih}\}(v_h - r_h) +(1-\lambda)\mathbb{1}\{ p_{i}\leq \overline{p}_{il}\} (v_l - r_l) \big] 
	\end{aligned}
\end{equation}	

	Consistent with the previous sections, the first line in Expression (\ref{eq:compinfo:hm:ut4}) represent the expected material payoff that an agent $i$ choosing $(p_{ih},p_{il},\overline{p}_i)\in \mathbb{R}^3_+$ pitted against an individual  who plays $(p_{jh},p_{jl},\overline{p}_j)\in \mathbb{R}^3_+$ would obtain when in the role of a \textit{Seller}. It is the weighted average surplus that she would get if in the possession of a high quality object and the surplus corresponding to having a low quality item. Again, the weights are the probabilities of getting the high quality object or the low quality one. If $j$ acquires the good (that is, if $p_{iQ}\leq\overline{p}_j$, with $Q \in \{h,l\}$), she receives the price she set and looses an amount equivalent to the cost of the good, $r_Q$. Her total surplus from selling is then $p_{iQ}-r_Q$. In turn, if the aforementioned exchange does not occur, the payoff is zero. 
	
	The second line is the expected material surplus that the agent would obtain when in the role of a \textit{Buyer}. It is the weighted average of the surplus that she would get if paired with a high quality seller and that corresponding to being matched with a low quality seller. The weights are the probabilities of the object being of high quality ($\lambda$) or of low quality (1-$\lambda$). In each case, Player $i$'s surplus corresponding to her potential buyer role is equivalent to her valuation $v_Q$ for the good minus the price $p_{jQ}$ she paid if it is the case that Player $j$ sells her the object (which happens only if $p_{jQ}\leq\overline{p}_i$), while she obtains $0$ if not. 
	
	As in the previous sections, the equilibrium set of the bilateral trading game with quality uncertainty and \textit{homo oeconomicus} agents can be found by studying the contingent game between a \textit{Buyer} and a \textit{Seller}. I next define all the possible types of result that can arise in the contingent game in Definition \ref{def:adv}.
	
	\begin{definition}\label{def:adv}
			Consider the contingent market where $i$ is in the \textit{Seller} role and $j\neq i$ is in the \textit{Buyer} role. The contingent game's outcomes can be of the following types:
		\begin{itemize}
			\item Full trade: both qualities are traded. This outcome realizes if and only if $\max\{p_{ih},p_{il}\} \leq \bar p_{j}$.
			\item Low quality: only low quality is traded. This outcome realizes if and only if $p_{il}\leq \bar p_{j} <  p_{ih}$.
			\item High quality: only high quality is traded. This outcome realizes if and only if $p_{ih}\leq \bar p_{j} <  p_{il}$.
			\item No trade: no quality is traded. This outcome realizes if and only if $ \bar p_{j} < \min\{p_{ih},p_{il}\}$.
		\end{itemize}
	\end{definition}	
		
	The equilibria are all pairs of  profiles that can be formed with the elements belonging to the contingent game's equilibrium set. In order to make the reading more efficient, I first define two relevant quantities that will be used in this and the following sections. Firstly, the expected consumer valuation, which is the average of the high and low quality valuations weighed by their respective probabilities ($\lambda$ and 1- $\lambda$). 	
	\begin{equation}\label{ve}
		v_e \equiv \lambda v_h + (1-\lambda) v_l.
	\end{equation}
	
\noindent Secondly, the minimum value of $\lambda$ for which the expected consumer valuation is weakly above the high-quality cost ($r_h\leq v_e$):
	
	\begin{equation}\label{le}
\lambda_e \equiv \frac{r_h-v_l}{v_h-v_l}.
	\end{equation}

	I study here two cases of this adverse selection setting. The first one is when not only the exchange of the high quality object but also of the low quality item produces a positive surplus and thus, is socially desirable. The second one is when the trade of the low quality good results in a negative surplus and it is then efficient from a social standpoint to avoid it exchanging hands.
	
\subsection{Desirable trade} \label{subsec:incompinfo:quality}
	I present in Proposition \ref{prop:incompinfo:selfish:maineq} the equilibria for the case where all trade is socially desirable and agents are of the \textit{homo oeconomicus} type. The contingent market outcomes listed are those defined in Definition \ref{def:adv}.
	
	\begin{proposition}\label{prop:incompinfo:selfish:maineq} 
		Consider the adverse selection bilateral trade game with $\kappa = 0$ and assume $0<r_l<v_l<r_h<v_h$. The Nash equilibria in pure strategies are:
		\begin{itemize}
		\item If $\lambda \geq \lambda_e$, A/B where $A,B \in \{ \text{Full trade},\text{Low quality},\text{No trade}\}$.

		\item If $\lambda < \lambda_e$, A/B where $A,B \in \{\text{Low quality},\text{No trade}\}$
	
		\end{itemize}
		\end{proposition}

	When the expected buyer valuation $v_e$ is (weakly) above the high quality cost $r_h$ or, equivalently, $\lambda \geq \lambda_e$, there are six types of equilibria. The first kind are those where trade takes place no matter the role distribution or quality of the goods. Next, there are equilibria where Player $i$ sells to Player $j$ only when in possession of a bad quality item, but the latter sells when her good is either quality. The following class of equilibria has Player $i$ not selling to Player $j$ and the latter selling to Player $i$ no matter the quality. The fourth kind of equilibria are those where only trade of the bad quality item takes place in both contingent markets. The game also features equilibria where Player $i$ does not sell to Player $j$ but the latter sells when her good is of low quality. The last class of equilibria are those where trade does not take place at all. Notice that in this case, inefficient outcomes arise because of the coordination failure introduced by the simultaneous-action assumption, mirroring the complete information case.
	
	In turn, when $\lambda < \lambda_e$, the equilibria laid out in Proposition \ref{prop:incompinfo:selfish:maineq} are of three types. The first one is that where only the bad quality good is traded in both contingent markets. The second sort is the case where the bad quality good is traded in one contingent market but the other one features no trade. Finally, there are equilibria where no trade occurs in either contingent market. The expected consumer valuation $v_e$ being below high quality cost $r_h$ implies that the asymmetric information problem is much more harmful for market efficiency, as the high quality item is never traded, even though it is socially desirable to do so. This situation resembles the ``lemons'' problem, as high quality sellers are driven out of the market and, even when the coordination failure is avoided (and trade takes place) only bad quality sellers remain.
	
	
	While considering valuation uncertainty and \textit{homo oeconomicus} players brought about the possibility of equilibria where low-valuation consumers are not served (like in the traditional monopoly pricing problem), the introduction of quality uncertainty adds a different potential inefficiency, that of adverse selection (where only the bad quality item is traded). In fact, for sufficiently low expectations about product quality, the most efficient profiles (those featuring full trade) do not even belong to the equilibrium set. This renders the adverse selection setting more problematic than the heterogeneous buyers case in terms of efficiency, as with \textit{homo oeconomicus} agents the efficient equilibria are not even attainable for sufficiently low values of the $\lambda$.
	
	I now explore whether considering partially moral agents does away with the various inefficient equilibria found in the \textit{homo oeconomicus} version of the model. Proposition \ref{prop:incompinfo:fulltrade4} presents the bilateral trade game's equilibrium types when agents are \textit{homo moralis}. The detailed characterisation is relegated to Proposition \ref{prop43star} in the Appendix.
	
		\begin{proposition}\label{prop:incompinfo:fulltrade4}
			Consider the adverse selection bilateral trade game with $\kappa \in (0,1)$ and assume $0<r_l<v_l<r_h<v_h$. If $\lambda \geq \lambda_e$, the only Nash equilibria in pure strategies are of the type Full trade/Full trade. If $\lambda < \lambda_e$, there are no Nash equilibria in pure strategies. 
		\end{proposition}

\begin{figure}[H]
	\centering
	\begin{tikzpicture}[thick, scale=1.5]
	
	\begin{axis}[
	width=10cm,
	height=4cm,
	axis lines = middle,
	xtick={0,1.8},
	ytick={0,600},
	yticklabels={$0$, $1$},
	xticklabels={$0$, $1$},
	clip=false,
	scaled ticks=false,
	xlabel = {$\lambda$},
	ylabel = {$\kappa$},
	xmin=0, 
	xmax=2,
	ymin=0, 
	ymax=1000]

	\filldraw [name path = rect, fill=black, opacity = 1] (0.65*100,-1) rectangle (1*180,1);  
		\filldraw [name path = rect, fill=blue, opacity = 1] (0,-1) rectangle (0.65*100,1);      

		\filldraw [fill=gray, opacity = 0.2] (0.65*100,1) rectangle (1*180,60);    

	\addplot[thick, samples=1000, dashed,domain=0:600,black, name path=three] coordinates {(0.65,0)(0.65,600)}
	node[pos=-.2]{$\lambda_e$};
	
	\path[name path = axis] (axis cs:0,0) -- (axis cs:1,0);
	\node[right] at (80, -7) {\tiny $A,B \in \{\text{Full t.}, \text{Low q.}, \text{No t.}\}$};
	\node[right] at (83, 30) {\tiny Full t./Full t.};
	
	\node[right, blue] at (-0.5, -7) {\tiny $A,B \in \{\text{Low q.},$};
	\node[right, blue] at (-0.5, -15) {\tiny $\text{No t.}\}$};
	
	\node[right] at (20, 20) {\tiny No eq.};
	
	%
	%
	%
%
%
	%
	%
	\end{axis}
	
	\end{tikzpicture}
	\caption{Equilibrium regions in the $(\lambda,\kappa)$ space, with $\kappa \in [0,1)$.}
	\label{fig:equil_quality_uncert}
\end{figure}
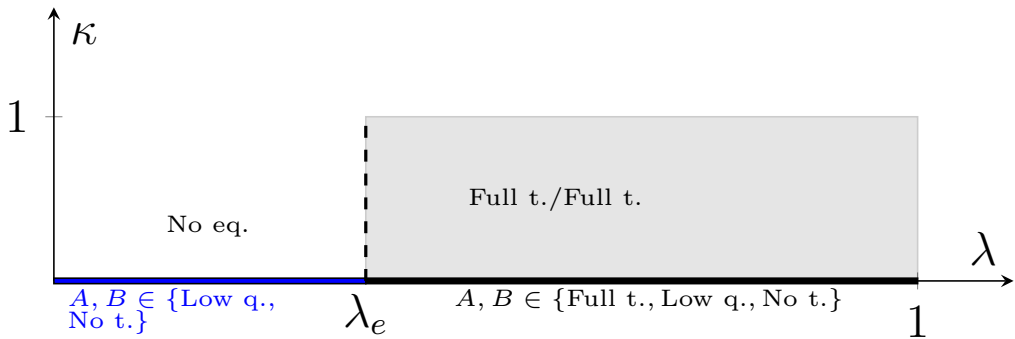

Proposition \ref{prop:incompinfo:fulltrade4} states that when the probability of the item being of good quality is large enough (or, equivalently, the high-quality cost $r_h$ is weakly below the expected consumption valuation $v_e$) the only Nash equilibria in  pure strategies are symmetric and feature full trade in both contingent markets. Furthermore, trade takes place at prices between the high-quality cost $r_h$ and the expected valuation $v_e$. Notice again, that these profiles are a subset of the equilibria in the \textit{homo oeconomicus} model.

Hence, in the case where agents are sufficiently ``optimistic'' about quality (so that $\lambda\geq \lambda_e$), going from a bilateral trade game where agents are endowed with \textit{homo oeconomicus} preferences to one where they are \textit{homo moralis} effectively eliminates all equilibria but the most materially efficient ones. In this context, these are the ones where both qualities are traded in both contingent markets. Thus, (partial) moral concerns not only prevent negotiation failures, as was already clear from the benchmark complete information model, but also eliminate adverse selection. Moreover and also in keeping with the findings for the previous variants of the model (complete information and heterogeneous buyer valuations), \textit{homo moralis} preferences result in a completely \textit{ex-ante} egalitarian result, as trade takes place in both contingent markets at the same price. 


The intuition behind the result that profiles with no trade of either quality in at least one contingent market are not equilibria when agents are \textit{homo moralis} is the same as for the full information benchmark. Starting from a situation with none or partial trade, agents can adjust their strategies so as not to modify the actual trade taking place but managing to set prices and thresholds such that the latter would not be above the former. Given that trade of both qualities is socially desirable, these deviations produce moral gains that are not offset by material losses and are thus profitable.

Secondly and also in line with the results for the full information case, asymmetric full trade profiles cannot be equilibria either. This is due to the fact that in a asymmetric profile with full trade in both contingent markets, there is at least one agent setting at least one price above her threshold. This provides an opportunity to make a moral gain and avoid material losses by simultaneously adjusting prices and thresholds. 

Finally, it is worth noting that when the the probability of having a high quality item is not high enough ($\lambda<\lambda_e$), there exist no Nash equilibria in pure strategies. This derives from the fact that \textit{homo moralis} will always find it profitable to deviate from a profile with full trade. The reason is that such profiles have either the \textit{Seller} trading at a price below cost or the the \textit{Buyer} paying above her expected valuation, so there are material gains to be had from deviating without sustaining moral losses. Agents will also be able to profitably deviate from profiles featuring contingent or no trade through moral gains.\footnote{The existence of a mixed-strategy equilibrium would be guaranteed by Nash's theorem if we consider a finite grid of prices and thresholds. Furthermore, if the grid is between and including costs and valuations for each quality and if it is fine enough, then all the pure strategy equilibria described in the paper remain so.}
%
	\subsection{Socially undesirable trade}\label{sec:inef}

In all of the above, I have assumed that costs are always lower than consumer valuations and thus that trading always produces a social surplus. In this Section I modify the model laid out in Section \ref{sec:incompinfo:quality} in order to reflect a situation where it would be socially inefficient to trade the low quality good. In this case, the sellers' valuation for the good ($r_l$) is actually higher than the buyers' ($v_l$). In line with this, I assume now that $0<v_l<r_l<r_h<v_h$. Proposition \ref{prop:incompinfo:selfish:maineqi} characterises the equilibria for the benchmark \textit{homo oeconomicus} case. Recall that quantities $v_e$ and $\lambda_e$ represent, respectively, the expected consumer valuation and the value of the probability of high quality above which the former is larger than the high quality cost $r_h$. They are defined in expressions (\ref{ve}) and  (\ref{le}).

\begin{proposition}\label{prop:incompinfo:selfish:maineqi} 
	Consider the adverse selection bilateral trade game with $\kappa = 0$ and assume that $0<v_l<r_l<r_h<v_h$. Then:
	
	\begin{enumerate}
	\item If $\lambda \geq \lambda_e$, the Nash equilibria in pure strategies of the bargaining game between two \textit{homo oeconomicus} are of the kind A/B where $A,B \in \{ \text{Full trade},\text{No trade}\}$
	\item If $\lambda < \lambda_e$, the Nash equilibria in pure strategies of the bargaining game are of the type No trade/No trade	

		\end{enumerate}
\end{proposition}

When the low quality object's cost is higher than its consumer valuation, the most efficient result is that only the high quality item be traded. However, Proposition \ref{prop:incompinfo:selfish:maineqi} shows that this type of profiles are not equilibria when agents are \textit{homo oeconomicus}. In fact, in a given contingent market, either both qualities are traded or non are, with the former equilibria only existing when $\lambda \geq \lambda_e$ (that is, the probability of the object being of high quality is large enough). This constitutes a relevant difference with the versions of the model analysed so far, as the most efficient profiles possible were in fact part of the equilibrium set when $\lambda \geq \lambda_e$. In addition, the case where $\lambda<\lambda_e$ leads to equilibria with no trade at all, whereas in the previous case it could also lead to exclusion of the high quality in addition to no-trade. How does the equilibrium set change when agents are of the \textit{homo moralis} type? Proposition \ref{prop:incompinfo:inef:hm:eq} provides the answer to that question. 

In order to facilitate the exposition, I define the following thresholds for $\lambda$ and $\kappa$:
$$\lambda_1 := \frac{r_h-r_l}{v_h-r_l}$$
%
%
$$\kappa_1:= \frac{r_h-r_l}{r_h-v_l}$$
%
%
%
$$\kappa_2(\lambda) := \min\left\{\frac{\lambda (v_h -v_l) + v_l - r_l}{\lambda (v_h - r_l) + r_l - v_l},\frac{\lambda (v_h - v_l) + v_l - r_h}{\lambda (v_h - r_l)+r_l - r_h}\right\}\text{, for }\lambda \in [\lambda_e,1]$$


Notice that $\lambda_1 < \lambda_e$. In addition, $\kappa_2(\lambda)$ is continuous and strictly increasing with $\kappa_2(\lambda_e) = 0$ and $\kappa_2(1) \in (\kappa_1,1)$. 

%
%
%

%

\begin{proposition}\label{prop:incompinfo:inef:hm:eq}
	Consider the adverse selection bilateral trade game with $\kappa \in (0,1)$ and assume that $0 < v_l < r_l < r_h < v_h$. The only kind of pure strategy Nash equilibria are:
		
	\begin{enumerate}
		\item 	High quality/High quality: if and only if $ \lambda \geq \lambda_1, \kappa \geq  \kappa_1$.
		\item 	Full trade/Full trade: if and only if $ \lambda \geq \lambda_e, \kappa \leq  \kappa_2(\lambda)$.
		\item   High quality/Full trade or Full trade/High quality: if and only if $ \kappa \in [\kappa_1,\kappa_2(\lambda)]$
		\end{enumerate}
%
%
%
%
%
%
%
%
%
\end{proposition}

Figure \ref{fig:hm_inef} summarises propositions \ref{prop:incompinfo:selfish:maineqi} and \ref{prop:incompinfo:inef:hm:eq}. The first relevant point to be made about Proposition \ref{prop:incompinfo:inef:hm:eq} is that it now takes a sufficiently high degree of morality to do away with all but the most efficient equilibria. This contrasts sharply with the conclusions for the preceding sections, where only efficient equilibria are attained for \textit{any} $\kappa \in (0,1)$. Even though a requirement on expected quality persists, sufficiently moral agents require a lower probability of good quality to reach these equilibria. In the same line, the existence of pure strategy Nash equilibria with partially moral agents now requires a lower expectation of good quality than in the previous case, provided that agents are moral ($\kappa\geq \kappa_1$) enough, since $\lambda_1<\lambda_e$.   

Considering partially moral agents brings about additional interesting changes to the equilibrium set of the game under study. In contrast with the previous cases where all trade was socially desirable, under the present parametrisation, \textit{homo moralis} agents reach equilibria that are not attainable by \textit{homo oeconomicus}. In other words, so far the introduction of partially moral concerns has had the effect of ``refining'' the equilibrium sets in the game between by \textit{homo oeconomicus} players, as every equilibrium profile with \textit{homo moralis} belonged to the former. This is not the case now because any profile with trade of exclusively the good quality in either or both contingent markets does not constitute an equilibrium with \textit{homo oeconomicus}.

This kind of equilibrium is particularly relevant, as only surplus-generating exchanges are undertaken. They exist when the agents' degree of morality and expectation about product quality are sufficiently high (i.e., the former is above $\kappa_1$ while the latter is above $\lambda_1$). The simple intuition behind such a result is that, given that trade is occurring at a price above high-quality cost $r_h$ and thus above low-quality cost $r_l$, only sufficiently moral agents will not be tempted to reduce their low quality price to sell their low quality item. It is worth pointing out that the minimum probability for the object to be of good quality required for the existence of this type of equilibria is strictly below that required for any trade to take place between \textit{homo oeconomicus}, or in other words ($\lambda_1<\lambda_e$).

The second type of equilibrium found in the \textit{homo moralis} model is Full trade/Full trade. This is the only class already present in the \textit{homo oeconomicus} case. It exists when the agents' degree of morality is not ``too high'' (below $\kappa_2(\lambda)$) and their expected consumer valuation is above the high quality cost, just as the \textit{homo oeconomicus} case. The fact that they exist only for $\lambda\geq\lambda_e$ is also in line with the result in the previous Section, where I assumed that all trade is socially desirable. The difference is that now there is an upper bound on the degree of morality ($\kappa_2(\lambda)$) for which the equilibria exist. This is because high degrees of morality now cause deviations that impede trade of the low quality to become profitable.

The third type are those with only the good quality being exchanged in one contingent market and full trade in the other one. The minimum degree of morality for which they exist is the same as the High quality/High quality equilibria, but its maximum is also bounded from above by $\kappa_2(\lambda)$, exactly like Full trade/Full trade profiles. This type of equilibria exist in the intersection of the sets of $\lambda$ and $\kappa$ values where the first two types are possible. That means that for a given $(\lambda,\kappa)$ vector that supports High quality/Full trade equilibria, High quality/High quality and Full trade/Full trade also exist. As reported in Proposition \ref{prop45star} in the Appendix, the price at which exchanges take place is exactly the same in both contingent markets and moreover, it is unique and increasing in the degree of morality. Intuitively, this is because more moral players who are selling both items will require a larger price not to deviate and cease the trade of the bad quality object. Simultaneously, for those selling only the good quality, they will be enticed to cease bad quality trade if the price is sufficiently small. The maximum price below which these deviations are profitable is increasing in $\kappa$.

Finally, it is worth pointing out the nonexistence of pure-strategy Nash equilibria. When

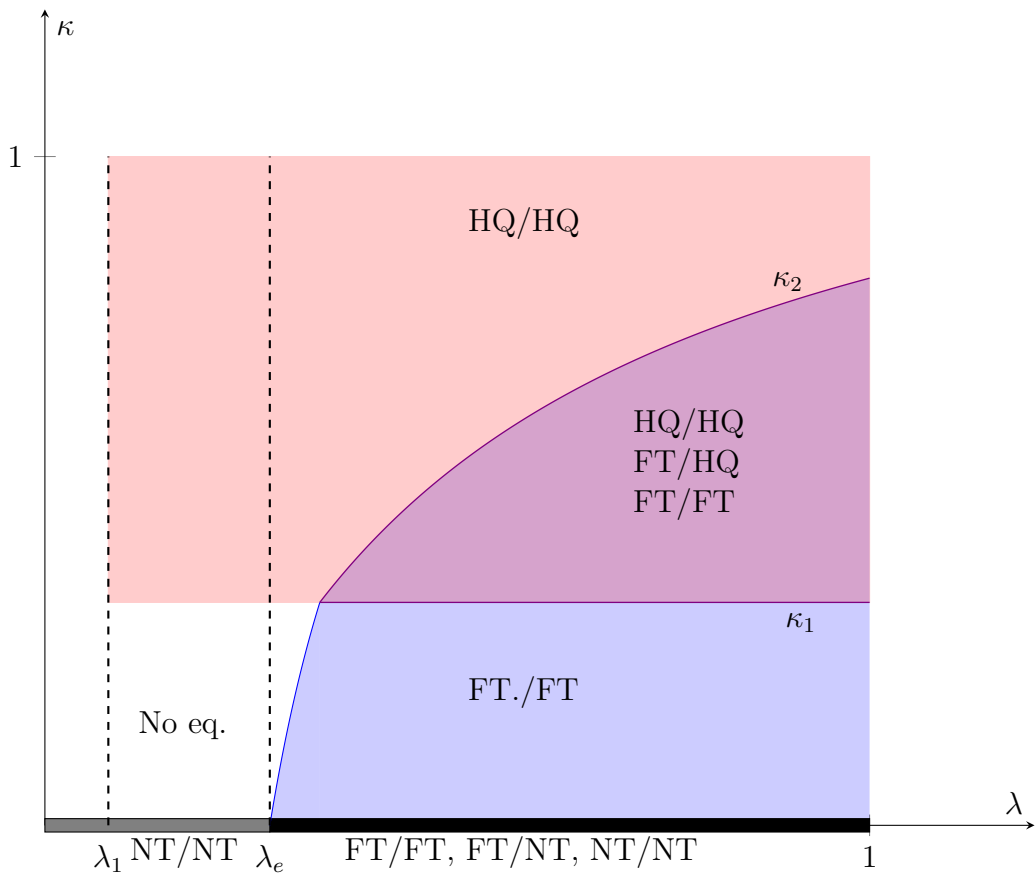
\begin{figure}[H]
	\centering
\begin{tikzpicture}[thick, scale=1.9]

\begin{axis}[
axis lines = middle,
xtick={0,1},
ytick={0,1},
clip=false,
xlabel = {$\lambda$},
ylabel = {$\kappa$},
xmin=0, 
xmax=1.2,
ymin=0, 
ymax=1.22]

\filldraw [name path=rect, pattern = dots, fill=red, draw=red, opacity=0.2] (1/13*100,1/3*100) rectangle (1*100,1*100);

\filldraw [name path = rect, fill=black, opacity = 1] (3/11*100,-1) rectangle (1*100,1);    
\filldraw [fill=gray] (0,-1) rectangle (3/11*100,1);    
%
\addplot[name path = A, color = violet, line width=0.5pt, 
domain = 1/3:1,
mark=none, 
samples=1000] {1/3}
node [very near end, below, black] {$\kappa_1$};
%
%
%
\addplot [name path = B, color = violet, line width=0.5pt,
domain = 1/3:1,
samples = 1000] {(x*55-10)/(x*45+10)}
node [very near end, above, black] {$\kappa_2$};
\addplot [name path = C,color = blue,
domain = 3/11:1/3,
samples = 1000] {(x*55-15)/(x*45-5)};
\addplot[thick, samples=1000, dashed,domain=0:1,black, name path=three] coordinates {(1/13,0)(1/13,1)}
node[pos=-0.05]{$\lambda_1$};
%
%
\addplot[thick, samples=1000, dashed,domain=0:1,black, name path=three] coordinates {(3/11,0)(3/11,1)}
node[pos=-0.05]{$\lambda_e$};
%
%
%
%
\path[name path=axis] (axis cs:0,0) -- (axis cs:1,0);
\addplot  fill between[
of= C  and axis,
soft clip={domain=0:1/3},
every segment no 0/.style={color=blue, opacity=0.2}];
];
\addplot fill between[
of= B  and axis,
soft clip={domain=1/3:1},
every segment no 0/.style={color=blue, opacity=0.2}];
%
%
%
\node[right] at (50, 90) {HQ/HQ};

\node[right] at (70, 60) {HQ/HQ};
\node[right] at (70, 54) {FT/HQ};
\node[right] at (70, 48) {FT/FT};

\node[right] at (50, 20) {FT./FT};
\node[right] at (9, -4) {NT/NT};
\node[right] at (35, -4) {FT/FT, FT/NT, NT/NT};
\node[right] at (10, 15) {No eq.};
%
%
\end{axis}

\end{tikzpicture}
			\caption{Equilibrium regions in the $(\lambda,\kappa)$ space. Shaded areas indicate types of equilibria for $\kappa \in [0,1)$.}
			\label{fig:hm_inef}
\end{figure}

	\section{A comparison with altruism}\label{sec:alt}

I have shown in the previous sections that (partially) Kantian agents manage to prevent negotiation failures and significantly reduce or even eliminate the inefficiencies stemming from information asymmetries. Cases where all potential exchanges produce a positive surplus feature particularly strong effects of morality on market outcomes: it suffices for agents to be partially moral for all but the most efficient profiles to cease being equilibria.

In this section I modify the previous models to consider agents with altruistic preferences. I describe as ``altruistic'' an individual who cares not only about her own material welfare but also about that of others. Both moralists and altruists deviate from the standard assumption that agents only care for their material payoff. Moreover, they may be behaviourally identical in certain contexts (see \cite{alger2016b}). However, the rationale for their observed choices is radically different. While the former care about the hypothetical material payoff they would get if their strategy is imitated by the other agent, altruists care about the material payoffs induced by their own strategy and their opponent's \textit{actual} strategy. Confusing morally-driven players and altruists could then lead to very different predictions when studying the reaction of agents to changes in the environment. 

An altruistic agent is partially concerned by the other player's material payoff. The utility function presented in (\ref{exanteutility}) when the \textit{ex-ante} game was defined needs to be modified accordingly. A standard way of introducing altruism is:
\begin{equation}\label{eq:compinfo:alt:utility}
U^{alt}\big((s_i,b_i),(s_j,b_j)\big)= \pi\big((s_i,b_i),(s_j,b_j))\big) + \alpha\cdot\pi\big((s_j,b_j),(s_i,b_i)\big), 
\end{equation}

\noindent where $\alpha \in[0,1]$ represents the degree of altruism. Equation (\ref{eq:compinfo:alt:utility}) states that the utility of an individual playing strategy $(s_i,b_i)$ against a player who employs strategy $(s_j,b_j)$ depends on her material payoff but is also positively affected by the material payoff obtained by her opposing party ($\pi\big((s_j,b_j),(s_i,b_i)\big)$ represents the payoff obtained by the player employing strategy $(s_j,b_j)$ against $(s_i,b_i)$). The more altruistic this player is (that is, the larger $\alpha$ is), the more utility she derives from her opponent's material well-being.\footnote{Although quite conventional, this representation it is not the only one. See for example \cite{bruhin2018} for an expression which allows for different degrees of altruism depending on whether the agent has a higher or a lower payoff than her opponent.} 

In the remainder of this Section I revisit the three games studied so far but I assume that agents are altruistic rather than morally concerned. The main conclusion is that in every setting considered, efficient results are attained only when altruism is high enough, and even a very high degree of altruism may not lead to full efficiency in some cases. In addition and echoing the \textit{homo oeconomicus} case, equilibria in the game between altruistic agents can be asymmetric. This implies that situations where one player has a larger \textit{ex-ante} utility 
%
%

In Proposition \ref{prop:compinfo:alt:maineq} I present the equilibria of the complete information game introduced in Section \ref{sec:compinfo} when agents have altruistic preferences. 


\begin{proposition}\label{prop:compinfo:alt:maineq} 
	Consider the complete information game with $\alpha \in (0,1)$. Then:
	\begin{itemize}
		\item If $\alpha \leq \frac{r}{v}$, the Nash equilibria in pure strategies are of the type A/B where $A,B \in \{\text{Trade},\text{No trade}\}$.	
		\item If $\alpha > \frac{r}{v}$, the Nash equilibria in pure strategies are of the type Trade/Trade	
	\end{itemize}
\end{proposition}	

When comparing the set of equilibria in the complete information game between altruistic agents with the \textit{homo oeconomicus} version, it is immediately clear that they are qualitatively similar for low levels of altruism ($\alpha \leq \frac{r}{v}$). In turn, for high levels of altruism the only possible equilibria are those of the Trade/Trade type. Therefore, a first conclusion is that only sufficiently high altruism restores efficiency in this game. An additional interesting point is that nothing prevents two equally altruistic agents to trade at different prices in the two contingent markets, a feature that the model with \textit{homo oeconomicus} also shares. What drives this result is the deontological nature of moral decision-making and more precisely, the fact that a player's decisions in a given contingent market are linked to her decisions in the other market through the moral Kantian term. This effect is not present when agents are consequentialistic (be it \textit{homo oeconomicus} or altruistic).

The above effects are a feature of the different versions of the bargaining game. However, considering information asymmetries brings about other interesting differences in the outcomes of the bargaining process between altruists when compared to the benchmark \textit{homo oeconomicus} and to \textit{homo moralis}. I now turn to the buyer heterogeneity game with altruistic agents. Proposition \ref{proposition:incompinfo:valunc:alt} lays out the equilibrium types for this version of the model.\footnote{The full characterisation is relegated to the Appendix because comparing prices and thresholds yields the same conclusion as in the complete information game.}
%
%
%
%
%
%
%
%

\begin{proposition}\label{proposition:incompinfo:valunc:alt} 
		Consider the buyer heterogeneity game between two agents with the same degree of altruism $\alpha \ (0,1)$. Then:

		\begin{enumerate}
			\item If $\alpha \leq \frac{r}{v_h}$, the Nash equilibria in pure strategies are of types A/B \\ where $A,B \in \{ \text{Full trade},\text{High valuation},\text{No trade}\}$.

				\item If $\frac{r}{v_h}>\alpha$, the Nash equilibria in pure strategies are of types A/B  \\ where $A,B \in \{ \text{Full trade},\text{High valuation}\}$.	
					
%
%
				
					
%
%
%

				%

		\end{enumerate}
\end{proposition}	

Proposition \ref{proposition:incompinfo:valunc:alt} shows that in the heterogeneous buyers setting, altruism does not restore full efficiency, as equilibria with exclusion of low valuation consumers are always possible for any $\alpha\in(0,1)$. Intuitively, this result originates in the fact that the altruistic seller actually perceives a lower cost when trading with a high valuation consumer.
  

I now turn to the adverse selection game in the case where it is socially desirable to trade both qualities. I define function $\lambda^{alt}_{e1}(\alpha)$  to be used in the propositions that follow. 
\begin{equation}
\lambda^{alt}_{e1}(\alpha) = \begin{cases} 
\dfrac{r_h - v_l + \alpha(r_l - v_h)}{v_h - v_l + \alpha(r_l - r_h)} & \text{ if } \alpha \leq \dfrac{r_{h}-v_{l}}{v_{h}-r_{l}} \\
0 & \text{otherwise }\\
\end{cases}
\end{equation}

For the case where $v_l<r_l$, I additionally define $\lambda^{alt}_{e2}(\alpha)$:

 \begin{equation}
\lambda^{alt}_{e2}(\alpha) = \begin{cases} 
\dfrac{r_l - v_l + \alpha(r_l - v_l)}{v_h - v_l + \alpha(r_l - r_h)} & \text{ if } \alpha \leq \dfrac{v_{h}-r_{l}}{r_{h}-v_{l}} \\
1 & \text{otherwise }\\
\end{cases}
\end{equation}

\noindent Notice that $\lambda^{alt}_{e1}(\alpha)$ is continuous and decreasing in $\alpha$. Moreover, $\lambda^{alt}_{e1}(0) =\lambda_e$, while $\lambda^{alt}_{e1}\left(\frac{r_{h}-v_{l}}{v_{h}-r_{l}}\right) =0$. In turn, function $\lambda^{alt}_{e2}(\alpha)$ is continuous and increasing in $\alpha$ for $v_l<r_l$, with $\lambda^{alt}_{e2}(0)=\frac{r_{l}-v_{l}}{v_{h}-v_{l}}$. 

Proposition \ref{proposition:incompinfo:qual:alt} presents the game's equilibrium types for different degrees of altruism when low quality trade is socially desirable ($r_l < v_l$), while Proposition \ref{proposition:incompinfo:altinef} does the same for the opposite case ($v_l<r_l$).

\begin{proposition}\label{proposition:incompinfo:qual:alt}
	Consider the adverse selection game between two agents with the same degree of altruism $\alpha \in (0,1)$ and assume $0<r_l < v_l < r_h < v_h$. Then, the Nash equilibria are of type $A/B\text{, where } A,B\in T(\alpha,\lambda)$ with:
	
	\begin{itemize}
		\item Full trade $\in T(\alpha,\lambda)\text{ if and only if } \lambda \geq \lambda^{alt}_{e1}(\alpha)$ 
		\item Low quality $\in T(\alpha,\lambda) \text{ if and only if } \alpha < \min \left\{\frac{r_h}{v_h},\frac{r_h-r_l}{v_h-v_l}\right\}$ 
		\item High quality $\in T(\alpha,\lambda) \text{ if and only if } \frac{r_h-r_l}{v_h-v_l}< \alpha<\frac{r_l}{v_l}$ 
		\item No trade $\in T(\alpha,\lambda) \text{ if and only if }  \alpha < \min \left\{\frac{r_h}{v_h},\frac{r_l}{v_l}\right\}$ 		
	\end{itemize}
\end{proposition}

\begin{proposition}\label{proposition:incompinfo:altinef} 
	Consider the adverse selection game between two agents with the same degree of altruism $\alpha \in (0,1)$ and assume $0<v_l<r_l<r_h<v_h$. Then: 
	
	\begin{itemize}
		\item Full trade $\in T(\alpha,\lambda) \iff \lambda \geq \max \{\lambda^{alt}_{e1}(\alpha),\lambda^{alt}_{e2}(\alpha)\}$ 
		\item High quality $\in T(\alpha,\lambda) \iff \alpha > \frac{r_h-r_l}{v_h-v_l}$	
		\item No trade $\in T(\alpha,\lambda) \iff  \alpha < \frac{r_h}{v_h}$ 
		
	\end{itemize}
	
\end{proposition}

When the exchange of all qualities is socially desirable, a sufficiently high degree of altruism does achieve full efficiency, eliminating all equilibria but those featuring full trade in both contingent markets. In contrast, the fully efficient result when there is some socially undesirable trade is not guaranteed by a very high degree of altruism. This is due to the fact that even though a sufficiently large $\alpha$ is indeed a necessary and sufficient condition for the existence of equilibria featuring only trade of the high quality good, it is not sufficient to eliminate equilibria with full trade in some contingent market.

\section{Discussion}\label{sec:disc}

In this paper I explore a simple bilateral trading process between two agents who may have moral concerns of a Kantian nature. In it, the agent in possession of the good being bargained over proposes a price, while the potential buyer simultaneously sets a threshold. The good only gets exchanged (at the proposed price) if the threshold is at least as high as the price. Importantly, I take an \textit{ex-ante} approach, analysing the game between two agents who face the same odds of finding themselves in either role. This framework is well suited to study the interaction between two agents who in addition to wondering ``what is the strategy that would leave me materially better-off?'', may also pose themselves the counter-factual question: ``how would I fare if my counter-party used the same strategy I am considering?''

My set-up does not rule out coordination failures and is flexible enough to allow for the study of the interaction between different kinds of information structures and moral preferences. I focus on two kinds of canonical information asymmetries. The first one originates when the \textit{Seller} is uncertain about the \textit{Buyer}'s valuation for the product, while the second one is the well known adverse selection problem where the \textit{Buyer} is not informed about the product's quality but the \textit{Seller} is. I tackle two versions of the latter: one where trade is always socially beneficial irrespective of the quality, and another where only trade of the highest quality is desirable.  
		
I show that considering moral preferences has important consequences on the outcome of the market under study. In situations where all trade is desirable (and irrespective of whether there is complete information, valuation uncertainty or adverse selection) and the version of the game populated by agents who are only concerned with their material payoff has efficient equilibria in addition to inefficient ones, partial morality eliminates the latter, leaving only the most efficient equilibria. It also results in a completely \textit{ex-ante} egalitarian outcome, in the sense that expected utility levels are the same. 

In the case with adverse selection and undesirable trade of the low quality items, when agents are of the \textit{homo oeconomicus} type, the most efficient equilibria (trade of only high quality goods) are not part of the equilibrium set. It takes a sufficiently high degree of morality for these profiles to become equilibria and an even larger one for them to be the sole kind of equilibrium possible.

The above contrasts sharply with another kind of pro-social preferences that have received a lot of attention in the literature: altruistic preferences. Indeed, I show that altruism is much less capable of eliminating inefficiencies stemming from information asymmetries. In particular, with altruistic agents, efficient equilibria remain as the sole equilibrium type only when information is equally distributed or in the case of adverse selection with desirable trade and, moreover, this is true only when the degree of altruism is high enough.

My results show that although moral preferences may go a long way in making the outcome of the bilateral trade process more efficient, they are not all-powerful. Notably, the most damaging incarnation of the adverse selection problem in this context (that in which trading low quality items actually reduces total surplus) is completely solved only for very high degrees of morality.

The latter not-withstanding, the inefficiencies introduced by the coordination problem as well as those originating in the information problem are strongly mitigated by considering moral preferences. This suggests that these kind of constraints might be less harmful than what standard models predict, especially in contexts well suited for the kind of partially deontological approach to decision-making that characterises \textit{homo moralis}.

%
	\nocite{smith2002}

 \newpage

\begin{appendices}




\section{Appendix to Section \ref{sec:compinfo}}

\noindent \textbf{Proof of Proposition \ref{proposition:compinfo:selfish:maineq}}
\begin{proof}[\unskip\nopunct]

The price set by a given player is completely independent of her threshold and only affects her utility through its interaction with her rival's threshold. Thus, the bargaining game's equilibrium set is all the pairs of profiles that can be formed from the equilibria of the contingent game. 

Consider now the contingent game. The \textit{Seller}'s and \textit{Buyer}'s best reply correspondences are:

	\[\bar p^{BR} = \argmax_{\bar p} U^{B}(\bar p;p)  = \begin{cases} 
	\overline{p}\geq p & \text{if } p\leq v \\
	\overline{p} < p & \text{if } p > v\\
	\end{cases}
	\] 
	
	\[ p^{BR} = \argmax_{p} U^{S}(p;\bar p) = \begin{cases} 
	p = 	\overline{p}  & \text{if } \overline{p} \geq r\\
	p>\overline{p}  & \text{if } \overline{p}< r\\
	\end{cases}
	\]  
	
	These correspondences intercept whenever $\overline{p}^*= p^*\in[r,v]$ or $\overline{p}^*<r, p^*>v$.
\end{proof}

%
The best-reply correspondences in the contingent game can be presented graphically in the $(p,\overline{p}) \in \mathbb{R}_+^2$ plane. I do so in Figure \ref{fig:bresponses:compinfo:selfish}. Notice that the contingent game has two types of distinct equilibrium profiles. The first are equilibria with trade, where the \textit{Buyer} sets a threshold in the $[r,v]$ interval and the \textit{Seller} proposes a price equal to the \textit{Buyer}'s threshold. The second are no-trade equilibria, with the \textit{Seller} choosing a price above the the \textit{Buyer}'s valuation $v$, who in turn sets a threshold below the \textit{Seller}'s cost $r$.

\begin{figure}[h]  
	\caption{Buyer and seller best-responses}\label{fig:bresponses:compinfo:selfish}
	
	\centering 
	\begin{tikzpicture}
	\draw[thick,<->] (0,7) node[above]{$\overline{p}$}--(0,0)--(8,0) node[right]{$p$};

	\node[inner sep=2pt,label=above right:{$\overline{p}=p$}] at (6,6) {};
	
	\draw [thick, red](4,4)--(6,6);
	\draw [thick, blue](0,0)--(2,2);
	\draw [very thick, black](2,2)--(4,4);

	\node[inner sep=2pt,label=left:{$r$}] at (0,2) {};
	\node[inner sep=2pt,label=left:{$v$}] at (0,4) {};
	\node[inner sep=2pt,label=below:{$r$}] at (2,0) {};
	\node[inner sep=2pt,label=below:{$v$}] at (4,0) {};

	\path[pattern=vertical lines,pattern color=blue] (0,0)--(0,6)--(4,4);
	
	\path[pattern=vertical lines,pattern color=blue] (0,6)--(4,4)--(4,6);
	
	\path[pattern=vertical lines,pattern color=blue] (4,0)--(4,4)--(6,6)--(6,0);
	
	\path[pattern=horizontal lines,pattern color=red] (0,0)--(2,2)--(6,2)--(6,0);
	
	\end{tikzpicture}
\end{figure}
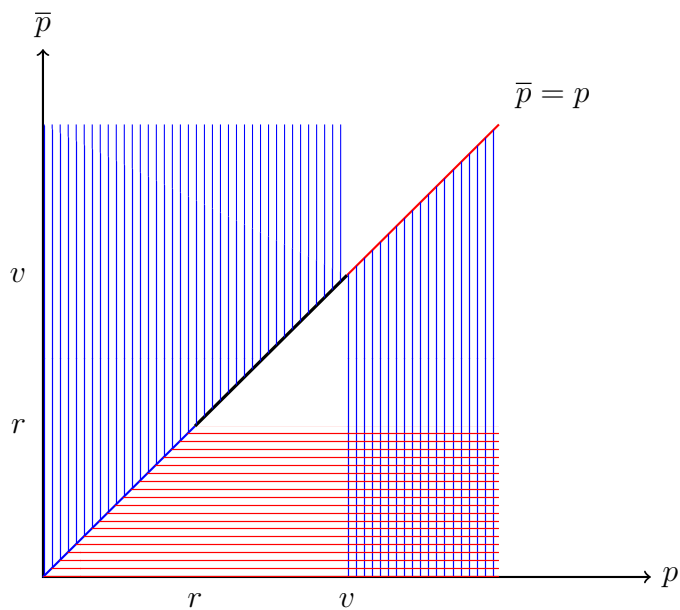	
\noindent \textbf{Proof of Proposition \ref{prop:compinfo:hm:eq}}
The proof follows directly from Lemmas \ref{lemma:compinfo:hm:contnotrade} to \ref{lemma:compinfo:hm:fulltrade3}, presented below.

\begin{lemma}\label{lemma:compinfo:hm:contnotrade}
	Consider the bilateral trade game with complete information and $\kappa\in(0,1)$. Profiles where at least one contingent market features no trade and its price is set below consumer valuation or its threshold above seller cost are not Nash equilibria in pure strategies.	
\end{lemma}

\begin{proof}
	Assume, without loss of generality, that $\bar p^*_1 < p^*_2\leq v$. Player 1 can profitably deviate by increasing her threshold to $\bar p'_1 = p^*_2$. The first term of her utility function --see Expression (\ref{eq:compinfo:hm:ut2})-- increases because $\bar p'_1 =  p^*_2\leq v$. The second term either increases (if $p^*_2 \geq p^*_1$) or remains the same (if $p^*_1 \leq \bar p^*_1$ or $p^*_2<\bar  p^*_1$).
	
	In turn, suppose $r\leq \bar p^*_2< p^*_1$. Player 1 can profitably deviate by reducing $p^*_1$ down to $p'_1=\bar p^*_2$. The first term of her utility function increases because $ p'_1 = \bar p^*_2\geq r$. The second term either increases (if $\bar p^*_1 \in [\bar p^*_2,p^*_1)$) or remains the same (if $\bar p^*_1 < \bar p^*_2$ or $\bar p^*_1\geq p^*_1$).
\end{proof}

\begin{lemma}\label{lemma:compinfo:hm:notrade}
	Consider the bilateral trade game with complete information and $\kappa\in(0,1)$. No-trade/No-trade profiles with prices above valuation and thresholds below costs are not Nash equilibria in pure strategies.	
	
\end{lemma}

\begin{proof}
	These profiles are described by $p^*_i>v, \overline{p}^*_j<r; $ for all $i,j \in\{1,2\},i\neq j$. Consider, without loss of generality, Player 1. She can profitably deviate by setting $\bar p'_1\in (\overline{p}^*_2,p^*_2)$ and $p'_1\in(\overline{p}^*_2,p'_1]$. The first term of her utility function remains unchanged while the second one increases.
\end{proof}
	
	

\begin{lemma}\label{lemma:compinfo:hm:NT_+_T}
	Consider the bilateral trade game with complete information and $\kappa\in(0,1)$. Profiles featuring no trade in one contingent market and trade in the other one are not Nash equilibria in pure strategies.
\end{lemma}
 \begin{proof}
 By Lemma \ref{lemma:compinfo:hm:contnotrade}, we can restrict attention to profiles where the price in the contingent market with no trade is set above $v$, while the threshold is below $r$. Assume, without loss of generality, that $\bar p_1 <r$ and $ p_2> v$. 
 
 Suppose, in addition, that $p_1 \leq \bar p_2, p_1 \in [r,v]$. Player 1 can then profitably deviate to $\bar p_1'= p_1$, as the first term of her utility remains unchanged while the second one unambiguously increases. 
 
 Suppose in turn, that $p_1 \leq \bar p_2<r$. Player 1 can profitably deviate to $p_1' \leq \bar p_1'\in[r,v]$, as the first term of her utility strictly increases while the second one  does not decrease. 
 
 Suppose alternatively, that $p^*_1 \leq \bar p_2<r$. Player 2 can profitably deviate to $ \bar p_2'=p_2'\in(r,v]$, as the first term of her utility strictly increases while the second one  does not decrease. 
 \end{proof}




\begin{corollary}\label{corollary:compinfo:hm:notrade}
	From Lemmas \ref{lemma:compinfo:hm:contnotrade}, \ref{lemma:compinfo:hm:notrade} and \ref{lemma:compinfo:hm:NT_+_T}, profiles featuring no trade in at least one contingent market are not equilibria of the bilateral trade game where at least one agent's degree of morality is $\kappa\in(0,1)$. 
\end{corollary}	

\begin{lemma}\label{lemma:compinfo:hm:fulltrade1}
	Consider the bilateral trade game with complete information and $\kappa\in(0,1)$. Full trade with the same price and threshold in each contingent market is a necessary condition for equilibrium.
\end{lemma}

\begin{proof}	
	Suppose first that $p_1^* < \bar p_2^*$. Then, Player 1 can profitably deviate to $p_1' = \bar p_2^*, \bar p_1' = \max\{\bar p_1^*,\bar p_2^*\}$. The first term of her utility function increases while the second one does not decrease.
 In turn, assume that $p^*_2 =  \bar p^*_1 < p^*_1 =  \bar p^*_2$. Then, Player 1 can profitably deviate to $\bar p_1' = p^*_1$. The first term of her utility function does not decrease while the second one increases.
\end{proof}	


%
%

\begin{lemma}\label{lemma:compinfo:hm:fulltrade3}
	Consider the bilateral trade game with complete information and $\kappa\in(0,1)$. Of all symmetric profiles with full trade only those with coordinates inside the $[r,v]$ interval are Nash equilibria in pure strategies.
\end{lemma}

\begin{proof}
From Lemma \ref{lemma:compinfo:hm:fulltrade1}, we can focus on symmetric, full-trade profiles where prices equal thresholds.

	Consider a profile where $p = p_1= \bar p_2= p_2=\bar p_1<r$. Player 1 can profitably deviate to $p'_1= p'_1>p$. The first term of her utility function increases while the second one remains the same.
	
	Consider next a profile where $v<p = p_1= \bar p_2= \bar p_1=p_2$. Then, Player 1 can profitably deviate to $\bar p'_1 = \bar p'_1 = p-\epsilon$. For small enough $\epsilon$, the first term of her utility function increases. The second term remains the same.
	
	Take a profile where $r\leq p^* = p^*_1= \bar p^*_2= \bar p^*_1=p^*_2\leq v$. Consider, without loss of generality, Player 1. Given Player 2's strategy, both terms of her utility function are at their maximum. Therefore, there are no profitable deviations possible. Then, the profile is an equilibrium. 
\end{proof}

\begin{proposition}\label{proposition:compinfo:hk:maineq}
	Consider the bilateral trade game with complete information and $\kappa=1$. The set of Nash equilibria in pure strategies is $p^*_i\leq \overline{p}^*_i \forall i \in \{1,2\}$
\end{proposition}

\begin{proof}
	Given that $r<v$, $i$'s utility function with $\kappa = 1$ is maximised if and only if $p^*_i\leq \overline{p}^*_i$.
\end{proof}	

Agents with degree of morality $\kappa$ equal to 1 (\textit{homo kantiensis}) will only care about her own strategy in order to compute their utility. That is, they will look for a best response against their own strategy and completely disregard what the other player might do. 

Under the parametrisation assumed, the surplus generated by trade is always positive. Therefore, agents are then better off when they propose a price and a threshold such that trade would take place if the other agent were to (hypothetically) post exactly the same ones. Recall, however, that the occurrence or not of trade in this model still depends on whether Player $1$'s threshold is at least as high as the price proposed by Player $2$ and vice-versa. 

A game between two \textit{homo kantiensis} will thus feature equilibria of two types. The first kind are full trade equilibria, while in the second one only contingent trade will be attained.

These equilibria will always feature some trade, either contingent or full. This can be easily shown by noticing that it is always be true that either $\overline{p}^*_2 \leq \overline{p}^*_1$ or $\overline{p}^*_2>\overline{p}^*_1$. In the first case, this implies that $p^*_2\leq \overline{p}^*_1 $ (in equilibrium it is necessarily true that $p^*_2 \leq\overline{p}^*_2\leq \overline{p}^*_1$). This means that Player 2 is selling to Player 1. If, in addition, it is also true that $\overline{p}_2^* \geq p^*_1$, then these profiles feature full trade, but this is not a necessary condition for equilibrium.

When comparing this game with the one between two \textit{homo oeconomicus}, an immediate conclusion is that the coordination failure introduced by the simultaneity in price and threshold decisions is mitigated, as no-trade equilibria are no longer possible. Notwithstanding, full morality does not always produce the more materially efficient result. This is because the game may still feature equilibria where trade is only contingent. These are a direct effect of agents only caring about their own threshold and prices while completely disregarding those of their rival. The result thus points to the fact that partially consequentialistic reasoning keeps \textit{homo moralis} ``anchored'' in reality by actually paying attention to her opponent's choices and thus producing actual trade.

The existence of completely deontological agents who only care about what would be ``the right thing to do'' might seem a little far fetched. However, the above result serves to illustrate a broader point already made in the preceding literature (notably in \cite{alger2017}): a higher degree of morality does not always imply a materially more efficient outcome, and whether it does is very context-dependent. In this particular case, the inefficiency arises because \textit{homo kantiensis} looses her grip on reality in the sense that she is not concerned about whether actual trade occurs, for which it is necessary that players set thresholds at least as high as their rival's prices. In contrast, she is content with just setting a price that she herself would accept, no matter how large or small.

\section{Appendix to Section \ref{sec:incompinfo:valunc}}	
\begin{customprop}{3.2*}\label{prop32star}

Consider the bilateral trade game with valuation uncertainty and $\kappa = 0$. The pure strategy Nash equilibria are:

\begin{itemize}
	\item Full trade/Full trade: $\forall i,j \in \{1,2\}, i\neq j : r \leq p^*_i= \min \{\bar p_{jh}^{*},\bar p_{jl}^{*}\}\leq v_l,\bar p_{ih}^{*} \in [\left(1-\lambda\right)p_{il}^{*}+\lambda r,\left(\frac{1}{\lambda}\right)p_{il}^{*}-\left(\frac{1-\lambda}{\lambda}\right)r]$
	
	\item $A/B$ with $A,B \in \{ \text{Full trade},\text{High valuation}\}, A\neq B$: for $i,j \in \{1,2\},i\neq j$:
	
	\begin{itemize}					
		\item $r\leq p^*_i= \min \{\bar p_{jh}^{*},\bar p_{jl}^{*}\}\leq v_l,\bar p_{ih}^{*} \in [\left(1-\lambda\right)p_{il}^{*}+\lambda r,\left(\frac{1}{\lambda}\right)p_{il}^{*}-\left(\frac{1-\lambda}{\lambda}\right)r];$
		
		\item $v_l<p^*_j=\bar p_{ih}^{*}\leq v_h, 0 \leq \bar p^{l*}_i \leq \lambda \bar p_{ih}^{*} + (1-\lambda) r$
	\end{itemize}

	\item $A/B$ with $A,B \in \{ \text{Full trade},\text{No trade}\}, A\neq B$: for $i,j \in \{1,2\},i\neq j$: 
	
	\begin{itemize}									
		\item $r\leq p^*_i= \min \{\bar p_{hj}^*,\bar p_{lj}^*\}\leq v_l, \bar p_{hi}^* \in [\left(1-\lambda\right) \bar p_{li}^* +\lambda r,\left(\frac{1}{\lambda}\right)\bar p_{li}^*-\left(\frac{1-\lambda}{\lambda}\right)r];$
		
		\item $\max\left\{\bar p_{li}^*,\bar p_{hi}^*\right\}<r, p^*_j>v_h$
	\end{itemize}

	\item $A/B$ with $A,B \in \{ \text{High valuation},\text{No trade}\}, A\neq B$: for $i,j \in \{1,2\},i\neq j$:
	
	\begin{itemize}				
		\item $\max\left\{\bar p_{lj}^*,\bar p_{jh}^*\right\}<r, p^*_i>v_h$
		
		\item $v_l<p^*_j=\bar p_{hi}^*<v_h, 0 \leq \bar p_{li}^* \leq \lambda \bar p_{hi}^* + (1-\lambda) r$
	\end{itemize}
	
	\item High valuation/High valuation: $\forall i,j \in \{1,2\}, i\neq j : v_l<p^*_j=\bar p_{hi}^*<v_h, 0 \leq \bar p_{li}^* \leq \lambda \bar p_{hi}^* + (1-\lambda) r$			

	\item No-trade in both contingent markets: $\forall i,j \in \{1,2\}, i\neq j : \max\left\{\bar p_{lj}^*,\bar p_{hj}^* \right\}<r, p^*_i>v_h$
\end{itemize}
\end{customprop}

\begin{proof}
	The price set by a given player is completely independent of her thresholds and only affects her utility through its interaction with her rival's thresholds. Thus, the bilateral trade game's equilibrium set is all the pairs of profiles that can be formed from the equilibria of the contingent game. 
	
	Consider now the contingent game where the seller sets price $p$ and the buyer sets thresholds $\bar p_l$ and $\bar p_h$. The buyer's best reply correspondences in the contingent game are:

	\begin{equation*}
	\bar p^{h}_{BR}= \argmax_{\bar  p^{h}} U^{B}(	\bar p^{l} ,\bar p^{h}; p) = \begin{cases} \label{eq_buyer_h_br}
	\bar p^{h}_{BR}= p & \text{if }  p \leq v_h \\
	\bar p^{h}_{BR}= p & \text{if }  p > v_h \\
	\end{cases}
	\end{equation*}
	
	\begin{equation*}
	\bar p^{l}_{BR}= \argmax_{\bar  p^{l}} U^{B}(	\bar p^{l} ,\bar p^{h}; p) = \begin{cases} \label{eq_buyer_l_br}
	\bar p^{l}_{BR}= p & \text{if }  p \leq v_l \\
	\bar p^{l}_{BR}= p & \text{if }  p > v_l \\
	\end{cases}
	\end{equation*}
	
	The seller's best reply correspondence in the contingent game is:	\begin{equation*}
	\begin{cases}
	p^{BR} > \max\{\bar p^{h},\bar p^{l}\} & \text{if }\max\{\bar p^{h},\bar p^{l}\}<r \\
	
	p^{BR} = \max\{\bar p^{h},\bar p^{l}\} & \text{if }\max\{\bar p^{h},\bar p^{l}\}\geq r, \min\{\bar p^{h},\bar p^{l}\}< r\\
	
	p^{BR} = \min\{\bar p^{h},\bar p^{l}\} & \text{if } \bar p^{l}\geq r, (1-\lambda)\bar p^{l}+\lambda r \leq \bar p^{h} \leq \frac{1}{\lambda}\bar p^{l}-\frac{1-\lambda}{\lambda}r\\
	
	p^{BR} = \bar p^{l} & \text{if } \bar p^{l}\geq r, \bar p^{h} <  (1-\lambda)\bar p^{l}+\lambda r\\
	
	p^{BR} = \bar p^{h} & \text{if } \bar p^{l}\geq r, \bar p^{h} >  \frac{1}{\lambda}\bar p^{l}-\frac{1-\lambda}{\lambda}r\\
	\end{cases}
	\end{equation*}
	
	They intersect at the following profiles:
		\begin{itemize}
		\item Both consumer types are served: $r \leq p^*=\min\{\bar p_{l},\bar p_{h}\} \leq v_l; \bar p_{l}^*\leq \bar p_{h}^*\leq \frac{\bar p_{l}^*}{\lambda}-\frac{1-\lambda}{\lambda}r $ or $\frac{\bar p_{h}^*-\lambda r}{1-\lambda} \leq \bar p_{h}^*\leq \bar p_{l}^*$		
		
		\item Only high valuation consumers are served: $v_l<p^*=\bar p_{h}^*\leq v_h,  \bar p_{l}^* \leq \lambda \bar p_{h}^* + (1-\lambda) r$
		
		\item No consumers are served: $\max\left\{\bar p_{l}^*,\bar p_{h}^* \right\}<r, p^*>v_h$
	\end{itemize}
%
%
%

\end{proof}	
\begin{customprop}{3.3*}\label{prop33star}
	Consider the bilateral trade game with heterogeneous buyer valuations and $\kappa \in (0,1)$. The Nash equilibria in pure strategies are $\bar p^*_{hi} \in [\left(1-\lambda\right)p^*_{li}+\lambda r,\left(\frac{1}{\lambda}\right)p^*_{li}-\left(\frac{1-\lambda}{\lambda}\right)r] \forall i,j \in \{1,2\}, i\neq j$.
\end{customprop}

\begin{proof}
The proof follows directly from Lemmas \ref{lemma:incompinfo:valunc:hm:1} to \ref{lemma:incompinfo:valunc:hm:10}, presented below.
\end{proof}

\begin{lemma}\label{lemma:incompinfo:valunc:hm:1}
	Consider the bilateral trade game with heterogeneous buyer valuations and $\kappa \in (0,1)$. Profiles where at least one of the contingent markets features no-trade with the high (low) valuation consumer and price below her valuation are not Nash equilibria in pure strategies.
\end{lemma}

\begin{proof}	
	Assume, without loss of generality, that $\bar p_{1Q} <p_2\leq v_Q$ for some $Q\in\{h,l\}$. Then, Player 1 can profitably deviate by setting $\bar p_{1Q}' = p_2$. Her utility function's first term increases while the second one does not decrease. 
\end{proof}	
\begin{lemma}\label{lemma:incompinfo:valunc:hm:2}
	Consider the bilateral trade game with heterogeneous buyer valuations and $\kappa \in (0,1)$. Profiles where at least one of the contingent markets features no-trade and a threshold above cost are not Nash equilibria in pure strategies.
\end{lemma}

\begin{proof}	
	Assume, without loss of generality, that $r \leq \max\{\bar p_{h2},\bar p_{l2}\}<p_1$. Then, Player 1 can profitably deviate by setting $p_1^{'}=\max\{\bar p_{h2},\bar p_{l2}\}$. Her utility function's first term increases while the second one does not decrease. 
\end{proof}	
\begin{lemma}\label{lemma:incompinfo:valunc:hm:3}
	Consider the bilateral trade game with heterogeneous buyer valuations and $\kappa \in (0,1)$. Profiles featuring no-trade in both contingent markets are not Nash equilibria in pure strategies.
\end{lemma}	

\begin{proof}	
	By Lemmas \ref{lemma:incompinfo:valunc:hm:1} and \ref{lemma:incompinfo:valunc:hm:2}, we can focus only on profiles where thresholds are below cost and prices above valuation: $\forall i,j \in \{1,2\}, i\neq j: \max\left\{\bar p_{lj},\bar p_{hj}\right\}<r, p_i>v_h$. Then, without loss of generality, Player 1 can profitably deviate to $r<p_1'= \bar p_{h1}'<v_l$. Her utility function's second term increases while the first one does not decrease. 
\end{proof}	

\begin{lemma}\label{lemma:incompinfo:valunc:hm:4}
	Consider the bilateral trade game with heterogeneous buyer valuations and $\kappa \in (0,1)$. Profiles featuring no-trade and exclusion are not Nash equilibria in pure strategies.
\end{lemma}	

\begin{proof}	
	By Lemmas \ref{lemma:incompinfo:valunc:hm:1} and \ref{lemma:incompinfo:valunc:hm:2}, we can focus only on profiles where the contingent market with no trade features thresholds below cost and price above valuation. Suppose first that $\max\left\{\bar p_{l1},\bar p_{h1} \right\}<r, p_2>v_h$ and $\bar p_{Q2}  < p_1\leq \bar p_{R2} < r ;Q,R \in \{h,l\}, Q\neq R$. Notice that Player 1 is selling at a price below cost. She can profitably deviate to $r<\bar p^{l'}_1=\bar p^{h'}_1 = p_1'<v_l$. Her utility function's first term increases while the second one does not decrease. 
	
	Suppose alternatively that $\max\left\{\bar p_{l1},\bar p_{h1} \right\}<r, p_2>v_h$ but $\max\left\{\bar p_{Q2},r \right\}  < p_1 \leq \bar p_{R2} < p_2;Q,R \in \{h,l\}, Q\neq R$. Notice that this means that $\max\left\{\bar p_{l1},\bar p_{h1} \right\}<p_1$. Then, Player 1 can profitably deviate to $ \bar p_{h1}'=p_1$. Her utility function's second term increases while the first one does not decrease. 
	
	Finally, assume that $\max\left\{\bar p_{l1},\bar p_{h1} \right\}<r, p_2>v_h$ and $\bar p_{Q2} < p_1,p_2\leq p_1\leq \bar p_{R2};Q,R \in \{h,l\}, Q\neq R$. Notice that this means that Player 2 is buying at a price above the highest valuation $v_h$. She can deviate to $\bar p_{Q2} \leq \bar p_{R2}' = p_2' = p_1-\epsilon,$ with $\epsilon \rightarrow 0$. She thus makes a material loss as a consequence of the price reduction. This loss can be made arbitrarily small. Then, her utility function's first term increases for a sufficiently small $\epsilon$ while the second one does not decrease. 
	
\end{proof}	
\begin{lemma}\label{lemma:incompinfo:valunc:hm:5}
	Consider the bilateral trade game with heterogeneous buyer valuations and $\kappa \in (0,1)$. Profiles featuring no-trade and full-trade are not Nash equilibria in pure strategies.
\end{lemma}	

\begin{proof}	
	By Lemmas \ref{lemma:incompinfo:valunc:hm:1} and \ref{lemma:incompinfo:valunc:hm:2}, we can focus only on profiles where the contingent market with no trade features thresholds below cost and price above valuation. Assume without loss of generality that $\max\left\{\bar p_{l1},\bar p_{h1} \right\}<r, p_2>v_h$ and $p_1\leq \min\left\{\bar p_{l2},\bar p_{l2} \right\}$.
	
	Suppose first that $p_1<r$. Then, Player 1 can profitably deviate to $r<p'_1 = \bar p_{l1}',\bar p_{h1}'< p_2$. Her utility function's first term increases while the second one does not decrease. 
	
	Alternatively, assume that $r\leq p_1\leq v_l$. Notice then that $\max\left\{\bar p_{l1},\bar p_{h1} \right\}<p_1<p_2$. Player 1 can profitably deviate by setting $\bar p_{l1}',\bar p_{h1}' = p_1$. She makes a moral gain and sustains no material losses.
	
	Finally, suppose that $p_1>v_l$. Player 2 can profitably deviate to $r<p'_2 = \bar p_{l2}'< p_1$. Her utility function's first term increases while the second one does not decrease. 
\end{proof}	
\begin{lemma}\label{lemma:incompinfo:valunc:hm:6}
	Consider the bilateral trade game with heterogeneous buyer valuations and $\kappa \in (0,1)$. Profiles featuring any exchange at a price above the high valuation are not Nash equilibria in pure strategies.
\end{lemma}	

\begin{proof}	
	Assume, without loss of generality, that $ v_h < p_2\leq \bar p_{1Q}$ for some $Q \in \{h,l\}$. Suppose first that $p_1<p_2$. Then, Player 1 can deviate to $\bar p_{1Q}' = p_1$. Her utility function's first term increases while the second one does not decrease. In turn, assume that $p_1>p_2$. If $p_2\leq \bar p_{1Q} < p_1$, then Player 1 can deviate to $p_1'=\bar p_{1Q}$ and obtain a moral gain without suffering material losses. If $p_2< p_1 \leq \bar p_{1Q}$,then it is Player 2 who can deviate by setting  $p_2'= \bar p_{2R} = \bar p_{1Q}-\epsilon$ for all $R \in \{h,l\}$, with $\epsilon<\bar p_{1Q}-p_2$. Her utility function's first term increases while the second one does not decrease. 
	
	Next, suppose that prices are equal, so $p_1=p_2$. Then, Player 1 can profitably deviate by setting $\bar p_{1Q}' =  p_1' = p_2-\epsilon$, with $\epsilon \rightarrow 0$. The deviation entails a strictly positive material gain from ceasing to buy at $p_2>v_h$ and a potential loss from selling at the lower price $p_2-\epsilon$ and no moral losses. Since the material loss can be made arbitrarily low by choosing a small enough $\epsilon$, the deviation is profitable. Her utility function's first term increases while the second one does not decrease. 
\end{proof}


\begin{lemma}\label{lemma:incompinfo:valunc:hm:7}
	Consider the bilateral trade game with heterogeneous buyer valuations and $\kappa \in (0,1)$. Profiles with at least some trading in at least one contingent market are not Nash equilibria if the price is below cost.
\end{lemma}	

\begin{proof}	
	Assume, without loss of generality that $p_1<r,p_1\leq \bar p_{2Q}$ for some $Q \in \{h,l\}$. Then, if $\bar p_{1R}< r$ for some $R \in \{h,l\}$, Player 1 can profitably deviate by setting $p_1'=\bar p_{1R}' = r$, whereas if $\bar p_{1R}>r$ for all $R \in \{h,l\}$, she can deviate to $p_1'= r$. Her utility function's first term increases while the second one does not decrease. 
\end{proof}	

\begin{lemma}\label{lemma:incompinfo:valunc:hm:8}
	Consider the bilateral trade game with heterogeneous buyer valuations and $\kappa \in (0,1)$. Profiles featuring exclusion in one contingent market and full trade in the other one are not Nash equilibria in pure strategies.
\end{lemma}	

\begin{proof}	
	Assume, without loss of generality that $\bar p_{Q2}<p_1<\bar p^{R2}_2$ for $Q,R \in \{h,l\}, Q\neq R$, and $p_2\leq \min \{\bar p^{h*}_1,\bar p^{l*}_1\}$. By Lemmas \ref{lemma:incompinfo:valunc:hm:6} and \ref{lemma:incompinfo:valunc:hm:7}, we can focus on profiles where $r\leq p_i^*\leq v_h \forall i \in \{1,2\}$. Then:
	
	(1) Profiles where $p_2^*<p_1^*$ are not equilibria. Indeed, Player 1 will always be able to deviate from a profile where $p_2^*<p_1^*$ and $\bar p^{Q*}_1<p_1^*$ for some $Q \in \{h,l\}$ to $\bar p^{Q'}_1\geq p_1^*$, as she obtains a moral gain and suffers no material losses. That only leaves  $p_2^*<p_1^* \leq \min \{\bar p^{h*}_1,\bar p^{l*}_1\}$. But then, it is Player 2 who can deviate to $p_2'=\bar p^{Q'}_2=p_1^*-\epsilon$, with $\epsilon \rightarrow 0$. Her utility function's first term increases while the second one does not decrease.  Therefore, we cannot have  $p_2^*<p_1^*$.
	
	(2) Profiles where $p_1^*\leq p_2^*$ are not equilibria. Indeed, Player 1 will always be able to deviate from a profile where $p_1^*\leq p_2^*$ and $v_l<p_2^*$ by setting $p_1'=\bar p^{l'}_2=p_2^*-\epsilon$, with $\epsilon \rightarrow 0$. Her utility function's first term increases while the second one does not decrease.  That only leaves $p_1^*\leq p_2^* \leq v_l$. But then, it is Player 2 who can deviate to $p_2'=\bar p^{Q'}_2=p_1^*$. Her utility function's first term increases while the second one does not decrease. Therefore, we cannot have  $p_1^*\leq p_2^*$.
\end{proof}	


\begin{lemma}\label{lemma:incompinfo:valunc:hm:9}
	Consider the bilateral trade game with with heterogeneous buyer valuations and $\kappa \in (0,1)$. Profiles featuring exclusion in both contingent markets are not Nash equilibria in pure strategies.
\end{lemma}	

\begin{proof}	
	Consider profiles where $\bar p^*_{Qi}<p_j^*<\bar p^*_{Rj}; \forall i,j \in \{1,2\}, i\neq j$ and $Q,R \in \{h,l\}, Q\neq R$. By Lemma \ref{lemma:incompinfo:valunc:hm:1}, we can focus only on the subset of profiles where $v_l<p_i^*\leq v_h  \forall i \in \{1,2\}$. Then:
	
	(1) Profiles where prices are different are not equilibria. Suppose, without loss of generality, that $p_1^* < p_2^*$. Then, if $\bar p^*_{h2} < p_2^*$, Player 2 can deviate to $\bar p_{h2}' = p_2^*$ to obtain a moral gain while suffering no material losses. Thus, we can only have $\bar p^*_{h2} \geq p_2^*$. But then, Player 1 finds it profitable to deviate to $p_1' = \bar p_{l1}' = p_2^* - \epsilon$, with $\epsilon < p_2^* - p_1^*$. Her utility function's first term increases while the second one does not decrease. 
	
	(2) But neither are equilibria profiles where prices are the same. Take $\max \{\bar p^*_{l1},\bar p^*_{l2}\}<p_1^*=p_2^*\leq  \min \{\bar p^*_{h1},\bar p^*_{h2}\}$. Then, without loss of generality, Player 1 can deviate to $p_1' = \bar p^{l'}_1 = p_2^*-\epsilon$, with $\epsilon \rightarrow 0$. She thus obtains a strictly positive moral gain that she trades off against a material loss that can be made arbitrarily small. Her utility function's second term increases while the first one does not decrease. 
	
\end{proof}	


\begin{lemma}\label{lemma:incompinfo:valunc:hm:10}
	Consider the bilateral trade game with heterogeneous buyer valuations and $\kappa \in (0,1)$. Profiles featuring pooling in both contingent markets are Nash equilibria only if $\bar p^*_{hi} \in [\left(1-\lambda\right)p^*_{li}+\lambda r,\left(\frac{1}{\lambda}\right)p^*_{li}-\left(\frac{1-\lambda}{\lambda}\right)r] \forall i,j \in \{1,2\}, i\neq j$.
\end{lemma}	

\begin{proof}		
	By Lemmas \ref{lemma:incompinfo:valunc:hm:6} and \ref{lemma:incompinfo:valunc:hm:7}, we can focus on profiles where prices are weakly above cost and below the high valuation. Take a profile with full trade in both contingent markets, so that $\min \{\bar p^*_{hi},\bar p^*_{li}\}<p_j^*$ and $r\leq p_j^*\leq v_h,\forall i,j\in \{1,2\}, i\neq j$.
	
	(1)	If prices are different, the profile is not a Nash equilibrium. To see it, assume without loss of generality that $p_1^* < p_2^*$. Then, if $\bar p^*_{Q2} < p_2^*$ for some $Q \in \{h,l\}$, Player 2 can profitably deviate to $\bar p_{Q2}' \geq p_2^*$. She obtains a moral gain without suffering any material losses. In turn, if $\bar p^*_{Q2} \geq p_2^*$ for all $Q \in \{h,l\}$ are not Nash equilibria either, as Player 1 can profitably deviate to $p_1' = p_2^*$, obtaining a material gain without experiencing moral losses (her utility function's first term increases while the second one does not decrease).   
	
	(2) If prices are equal but above the low valuation, the profile is not a Nash equilibrium. Indeed, suppose that $v_l < p_1^* = p_2^*$. Then, without loss of generality, Player 1 can deviate to $\bar p^{l'}_1 = p_1' = p_2^*-\epsilon$, with $\epsilon \rightarrow 0$. She obtains a strictly positive material gain by not buying from Player 2 in the low-valuation case. This gain outweighs the material loss brought about by the price reduction if $\epsilon$ is small enough. In addition, she suffers no moral losses. Her utility function's first term increases while the second one does not decrease. 
	
	(3) With equal prices, profiles where a player's lowest threshold is above the common price are not Nash equilibria. Indeed, suppose without loss of generality that $p_1^* = p_2^*<\min\{\bar p^*_{l2},\bar p^*_{h2}\}$. Then, Player 1 can deviate to $p_1' = \bar p_{h1}' = \bar p_{l1}'$ Her utility function's first term increases while the second one does not decrease. 	
	
	(4) With equal prices, if $(\bar p^*_{li},\bar p^*_{hi})$ are such that $\bar p^*_{hi} \notin [\left(1-\lambda\right)p^*_{li}+\lambda r,\left(\frac{1}{\lambda}\right)p^*_{li}-\left(\frac{1-\lambda}{\lambda}\right)r]$ for some $i \in \{1,2\}$ then the profile is not an equilibrium. Indeed, suppose without loss of generality that $\bar p^*_{h2} \notin [\left(1-\lambda\right)p^*_{l2}+\lambda r,\left(\frac{1}{\lambda}\right)p^*_{l2}-\left(\frac{1-\lambda}{\lambda}\right)r]$. Then, Player 1 can deviate to $p_1'= \bar p_{h1}' = \bar p_{l1}' = \max \{p^*_{l2},p^*_{l2}\}$. Her utility function's first term increases while the second one does not decrease. 
	
	(5) Profiles where $r\leq p_i^* = p_j^* \leq v_l$ and $\bar p^*_{hi} \in [\left(1-\lambda\right)p^*_{li}+\lambda r,\left(\frac{1}{\lambda}\right)p^*_{li}-\left(\frac{1-\lambda}{\lambda}\right)r] \forall i,j \in \{1,2\}, i\neq j$ are equilibria. Indeed, notice that given $(p_j^*,p^*_{lj},p^*_{hj})$, $(p_i^*,p^*_{li},p^*_{hi})$ maximises \ref{eq:incompinfo:valunc:payoff}. The strategies are then mutual best responses and therefore, these profiles are Nash equilibria.
	
\end{proof}


\begin{proposition}\label{proposition:valunc:hk:maineq}
	The set of Nash equilibria of the bilateral trade game with heterogeneous buyer valuations and $\kappa = 1$ is $p^*_i\leq \min \{\bar p^*_{hi},\bar p^*_{li}\} \forall i \in \{1,2\}$
\end{proposition}			

\begin{proof}
	Given that $r<v_l<v_h$, Expression (\ref{eq:compinfo:hm:ut3}) with $\kappa = 1$ is maximised if and only if $p^*_i\leq \min \{\bar p^*_{hi},\bar p^*_{li}\}$ .
\end{proof}	
\section{Appendix to Section \ref{sec:incompinfo:valunc}}	
\begin{customprop}{3.2*}\label{prop32star}

Consider the bilateral trade game with valuation uncertainty and $\kappa = 0$. The pure strategy Nash equilibria are:

\begin{itemize}
	\item Full trade/Full trade: $\forall i,j \in \{1,2\}, i\neq j : r \leq p^*_i= \min \{\bar p_{jh}^{*},\bar p_{jl}^{*}\}\leq v_l,\bar p_{ih}^{*} \in [\left(1-\lambda\right)p_{il}^{*}+\lambda r,\left(\frac{1}{\lambda}\right)p_{il}^{*}-\left(\frac{1-\lambda}{\lambda}\right)r]$
	
	\item $A/B$ with $A,B \in \{ \text{Full trade},\text{High valuation}\}, A\neq B$: for $i,j \in \{1,2\},i\neq j$:
	
	\begin{itemize}					
		\item $r\leq p^*_i= \min \{\bar p_{jh}^{*},\bar p_{jl}^{*}\}\leq v_l,\bar p_{ih}^{*} \in [\left(1-\lambda\right)p_{il}^{*}+\lambda r,\left(\frac{1}{\lambda}\right)p_{il}^{*}-\left(\frac{1-\lambda}{\lambda}\right)r];$
		
		\item $v_l<p^*_j=\bar p_{ih}^{*}\leq v_h, 0 \leq \bar p^{l*}_i \leq \lambda \bar p_{ih}^{*} + (1-\lambda) r$
	\end{itemize}

	\item $A/B$ with $A,B \in \{ \text{Full trade},\text{No trade}\}, A\neq B$: for $i,j \in \{1,2\},i\neq j$: 
	
	\begin{itemize}									
		\item $r\leq p^*_i= \min \{\bar p_{hj}^*,\bar p_{lj}^*\}\leq v_l, \bar p_{hi}^* \in [\left(1-\lambda\right) \bar p_{li}^* +\lambda r,\left(\frac{1}{\lambda}\right)\bar p_{li}^*-\left(\frac{1-\lambda}{\lambda}\right)r];$
		
		\item $\max\left\{\bar p_{li}^*,\bar p_{hi}^*\right\}<r, p^*_j>v_h$
	\end{itemize}

	\item $A/B$ with $A,B \in \{ \text{High valuation},\text{No trade}\}, A\neq B$: for $i,j \in \{1,2\},i\neq j$:
	
	\begin{itemize}				
		\item $\max\left\{\bar p_{lj}^*,\bar p_{jh}^*\right\}<r, p^*_i>v_h$
		
		\item $v_l<p^*_j=\bar p_{hi}^*<v_h, 0 \leq \bar p_{li}^* \leq \lambda \bar p_{hi}^* + (1-\lambda) r$
	\end{itemize}
	
	\item High valuation/High valuation: $\forall i,j \in \{1,2\}, i\neq j : v_l<p^*_j=\bar p_{hi}^*<v_h, 0 \leq \bar p_{li}^* \leq \lambda \bar p_{hi}^* + (1-\lambda) r$			

	\item No-trade in both contingent markets: $\forall i,j \in \{1,2\}, i\neq j : \max\left\{\bar p_{lj}^*,\bar p_{hj}^* \right\}<r, p^*_i>v_h$
\end{itemize}
\end{customprop}

\begin{proof}
	The price set by a given player is completely independent of her thresholds and only affects her utility through its interaction with her rival's thresholds. Thus, the bilateral trade game's equilibrium set is all the pairs of profiles that can be formed from the equilibria of the contingent game. 
	
	Consider now the contingent game where the seller sets price $p$ and the buyer sets thresholds $\bar p_l$ and $\bar p_h$. The buyer's best reply correspondences in the contingent game are:

	\begin{equation*}
	\bar p^{h}_{BR}= \argmax_{\bar  p^{h}} U^{B}(	\bar p^{l} ,\bar p^{h}; p) = \begin{cases} \label{eq_buyer_h_br}
	\bar p^{h}_{BR}= p & \text{if }  p \leq v_h \\
	\bar p^{h}_{BR}= p & \text{if }  p > v_h \\
	\end{cases}
	\end{equation*}
	
	\begin{equation*}
	\bar p^{l}_{BR}= \argmax_{\bar  p^{l}} U^{B}(	\bar p^{l} ,\bar p^{h}; p) = \begin{cases} \label{eq_buyer_l_br}
	\bar p^{l}_{BR}= p & \text{if }  p \leq v_l \\
	\bar p^{l}_{BR}= p & \text{if }  p > v_l \\
	\end{cases}
	\end{equation*}
	
	The seller's best reply correspondence in the contingent game is:	\begin{equation*}
	\begin{cases}
	p^{BR} > \max\{\bar p^{h},\bar p^{l}\} & \text{if }\max\{\bar p^{h},\bar p^{l}\}<r \\
	
	p^{BR} = \max\{\bar p^{h},\bar p^{l}\} & \text{if }\max\{\bar p^{h},\bar p^{l}\}\geq r, \min\{\bar p^{h},\bar p^{l}\}< r\\
	
	p^{BR} = \min\{\bar p^{h},\bar p^{l}\} & \text{if } \bar p^{l}\geq r, (1-\lambda)\bar p^{l}+\lambda r \leq \bar p^{h} \leq \frac{1}{\lambda}\bar p^{l}-\frac{1-\lambda}{\lambda}r\\
	
	p^{BR} = \bar p^{l} & \text{if } \bar p^{l}\geq r, \bar p^{h} <  (1-\lambda)\bar p^{l}+\lambda r\\
	
	p^{BR} = \bar p^{h} & \text{if } \bar p^{l}\geq r, \bar p^{h} >  \frac{1}{\lambda}\bar p^{l}-\frac{1-\lambda}{\lambda}r\\
	\end{cases}
	\end{equation*}
	
	They intersect at the following profiles:
		\begin{itemize}
		\item Both consumer types are served: $r \leq p^*=\min\{\bar p_{l},\bar p_{h}\} \leq v_l; \bar p_{l}^*\leq \bar p_{h}^*\leq \frac{\bar p_{l}^*}{\lambda}-\frac{1-\lambda}{\lambda}r $ or $\frac{\bar p_{h}^*-\lambda r}{1-\lambda} \leq \bar p_{h}^*\leq \bar p_{l}^*$		
		
		\item Only high valuation consumers are served: $v_l<p^*=\bar p_{h}^*\leq v_h,  \bar p_{l}^* \leq \lambda \bar p_{h}^* + (1-\lambda) r$
		
		\item No consumers are served: $\max\left\{\bar p_{l}^*,\bar p_{h}^* \right\}<r, p^*>v_h$
	\end{itemize}
%
%
%

\end{proof}	
\begin{customprop}{3.3*}\label{prop33star}
	Consider the bilateral trade game with heterogeneous buyer valuations and $\kappa \in (0,1)$. The Nash equilibria in pure strategies are $\bar p^*_{hi} \in [\left(1-\lambda\right)p^*_{li}+\lambda r,\left(\frac{1}{\lambda}\right)p^*_{li}-\left(\frac{1-\lambda}{\lambda}\right)r] \forall i,j \in \{1,2\}, i\neq j$.
\end{customprop}

\begin{proof}
The proof follows directly from Lemmas \ref{lemma:incompinfo:valunc:hm:1} to \ref{lemma:incompinfo:valunc:hm:10}, presented below.
\end{proof}

\begin{lemma}\label{lemma:incompinfo:valunc:hm:1}
	Consider the bilateral trade game with heterogeneous buyer valuations and $\kappa \in (0,1)$. Profiles where at least one of the contingent markets features no-trade with the high (low) valuation consumer and price below her valuation are not Nash equilibria in pure strategies.
\end{lemma}

\begin{proof}	
	Assume, without loss of generality, that $\bar p_{1Q} <p_2\leq v_Q$ for some $Q\in\{h,l\}$. Then, Player 1 can profitably deviate by setting $\bar p_{1Q}' = p_2$. Her utility function's first term increases while the second one does not decrease. 
\end{proof}	
\begin{lemma}\label{lemma:incompinfo:valunc:hm:2}
	Consider the bilateral trade game with heterogeneous buyer valuations and $\kappa \in (0,1)$. Profiles where at least one of the contingent markets features no-trade and a threshold above cost are not Nash equilibria in pure strategies.
\end{lemma}

\begin{proof}	
	Assume, without loss of generality, that $r \leq \max\{\bar p_{h2},\bar p_{l2}\}<p_1$. Then, Player 1 can profitably deviate by setting $p_1^{'}=\max\{\bar p_{h2},\bar p_{l2}\}$. Her utility function's first term increases while the second one does not decrease. 
\end{proof}	
\begin{lemma}\label{lemma:incompinfo:valunc:hm:3}
	Consider the bilateral trade game with heterogeneous buyer valuations and $\kappa \in (0,1)$. Profiles featuring no-trade in both contingent markets are not Nash equilibria in pure strategies.
\end{lemma}	

\begin{proof}	
	By Lemmas \ref{lemma:incompinfo:valunc:hm:1} and \ref{lemma:incompinfo:valunc:hm:2}, we can focus only on profiles where thresholds are below cost and prices above valuation: $\forall i,j \in \{1,2\}, i\neq j: \max\left\{\bar p_{lj},\bar p_{hj}\right\}<r, p_i>v_h$. Then, without loss of generality, Player 1 can profitably deviate to $r<p_1'= \bar p_{h1}'<v_l$. Her utility function's second term increases while the first one does not decrease. 
\end{proof}	

\begin{lemma}\label{lemma:incompinfo:valunc:hm:4}
	Consider the bilateral trade game with heterogeneous buyer valuations and $\kappa \in (0,1)$. Profiles featuring no-trade and exclusion are not Nash equilibria in pure strategies.
\end{lemma}	

\begin{proof}	
	By Lemmas \ref{lemma:incompinfo:valunc:hm:1} and \ref{lemma:incompinfo:valunc:hm:2}, we can focus only on profiles where the contingent market with no trade features thresholds below cost and price above valuation. Suppose first that $\max\left\{\bar p_{l1},\bar p_{h1} \right\}<r, p_2>v_h$ and $\bar p_{Q2}  < p_1\leq \bar p_{R2} < r ;Q,R \in \{h,l\}, Q\neq R$. Notice that Player 1 is selling at a price below cost. She can profitably deviate to $r<\bar p^{l'}_1=\bar p^{h'}_1 = p_1'<v_l$. Her utility function's first term increases while the second one does not decrease. 
	
	Suppose alternatively that $\max\left\{\bar p_{l1},\bar p_{h1} \right\}<r, p_2>v_h$ but $\max\left\{\bar p_{Q2},r \right\}  < p_1 \leq \bar p_{R2} < p_2;Q,R \in \{h,l\}, Q\neq R$. Notice that this means that $\max\left\{\bar p_{l1},\bar p_{h1} \right\}<p_1$. Then, Player 1 can profitably deviate to $ \bar p_{h1}'=p_1$. Her utility function's second term increases while the first one does not decrease. 
	
	Finally, assume that $\max\left\{\bar p_{l1},\bar p_{h1} \right\}<r, p_2>v_h$ and $\bar p_{Q2} < p_1,p_2\leq p_1\leq \bar p_{R2};Q,R \in \{h,l\}, Q\neq R$. Notice that this means that Player 2 is buying at a price above the highest valuation $v_h$. She can deviate to $\bar p_{Q2} \leq \bar p_{R2}' = p_2' = p_1-\epsilon,$ with $\epsilon \rightarrow 0$. She thus makes a material loss as a consequence of the price reduction. This loss can be made arbitrarily small. Then, her utility function's first term increases for a sufficiently small $\epsilon$ while the second one does not decrease. 
	
\end{proof}	
\begin{lemma}\label{lemma:incompinfo:valunc:hm:5}
	Consider the bilateral trade game with heterogeneous buyer valuations and $\kappa \in (0,1)$. Profiles featuring no-trade and full-trade are not Nash equilibria in pure strategies.
\end{lemma}	

\begin{proof}	
	By Lemmas \ref{lemma:incompinfo:valunc:hm:1} and \ref{lemma:incompinfo:valunc:hm:2}, we can focus only on profiles where the contingent market with no trade features thresholds below cost and price above valuation. Assume without loss of generality that $\max\left\{\bar p_{l1},\bar p_{h1} \right\}<r, p_2>v_h$ and $p_1\leq \min\left\{\bar p_{l2},\bar p_{l2} \right\}$.
	
	Suppose first that $p_1<r$. Then, Player 1 can profitably deviate to $r<p'_1 = \bar p_{l1}',\bar p_{h1}'< p_2$. Her utility function's first term increases while the second one does not decrease. 
	
	Alternatively, assume that $r\leq p_1\leq v_l$. Notice then that $\max\left\{\bar p_{l1},\bar p_{h1} \right\}<p_1<p_2$. Player 1 can profitably deviate by setting $\bar p_{l1}',\bar p_{h1}' = p_1$. She makes a moral gain and sustains no material losses.
	
	Finally, suppose that $p_1>v_l$. Player 2 can profitably deviate to $r<p'_2 = \bar p_{l2}'< p_1$. Her utility function's first term increases while the second one does not decrease. 
\end{proof}	
\begin{lemma}\label{lemma:incompinfo:valunc:hm:6}
	Consider the bilateral trade game with heterogeneous buyer valuations and $\kappa \in (0,1)$. Profiles featuring any exchange at a price above the high valuation are not Nash equilibria in pure strategies.
\end{lemma}	

\begin{proof}	
	Assume, without loss of generality, that $ v_h < p_2\leq \bar p_{1Q}$ for some $Q \in \{h,l\}$. Suppose first that $p_1<p_2$. Then, Player 1 can deviate to $\bar p_{1Q}' = p_1$. Her utility function's first term increases while the second one does not decrease. In turn, assume that $p_1>p_2$. If $p_2\leq \bar p_{1Q} < p_1$, then Player 1 can deviate to $p_1'=\bar p_{1Q}$ and obtain a moral gain without suffering material losses. If $p_2< p_1 \leq \bar p_{1Q}$,then it is Player 2 who can deviate by setting  $p_2'= \bar p_{2R} = \bar p_{1Q}-\epsilon$ for all $R \in \{h,l\}$, with $\epsilon<\bar p_{1Q}-p_2$. Her utility function's first term increases while the second one does not decrease. 
	
	Next, suppose that prices are equal, so $p_1=p_2$. Then, Player 1 can profitably deviate by setting $\bar p_{1Q}' =  p_1' = p_2-\epsilon$, with $\epsilon \rightarrow 0$. The deviation entails a strictly positive material gain from ceasing to buy at $p_2>v_h$ and a potential loss from selling at the lower price $p_2-\epsilon$ and no moral losses. Since the material loss can be made arbitrarily low by choosing a small enough $\epsilon$, the deviation is profitable. Her utility function's first term increases while the second one does not decrease. 
\end{proof}


\begin{lemma}\label{lemma:incompinfo:valunc:hm:7}
	Consider the bilateral trade game with heterogeneous buyer valuations and $\kappa \in (0,1)$. Profiles with at least some trading in at least one contingent market are not Nash equilibria if the price is below cost.
\end{lemma}	

\begin{proof}	
	Assume, without loss of generality that $p_1<r,p_1\leq \bar p_{2Q}$ for some $Q \in \{h,l\}$. Then, if $\bar p_{1R}< r$ for some $R \in \{h,l\}$, Player 1 can profitably deviate by setting $p_1'=\bar p_{1R}' = r$, whereas if $\bar p_{1R}>r$ for all $R \in \{h,l\}$, she can deviate to $p_1'= r$. Her utility function's first term increases while the second one does not decrease. 
\end{proof}	

\begin{lemma}\label{lemma:incompinfo:valunc:hm:8}
	Consider the bilateral trade game with heterogeneous buyer valuations and $\kappa \in (0,1)$. Profiles featuring exclusion in one contingent market and full trade in the other one are not Nash equilibria in pure strategies.
\end{lemma}	

\begin{proof}	
	Assume, without loss of generality that $\bar p_{Q2}<p_1<\bar p^{R2}_2$ for $Q,R \in \{h,l\}, Q\neq R$, and $p_2\leq \min \{\bar p^{h*}_1,\bar p^{l*}_1\}$. By Lemmas \ref{lemma:incompinfo:valunc:hm:6} and \ref{lemma:incompinfo:valunc:hm:7}, we can focus on profiles where $r\leq p_i^*\leq v_h \forall i \in \{1,2\}$. Then:
	
	(1) Profiles where $p_2^*<p_1^*$ are not equilibria. Indeed, Player 1 will always be able to deviate from a profile where $p_2^*<p_1^*$ and $\bar p^{Q*}_1<p_1^*$ for some $Q \in \{h,l\}$ to $\bar p^{Q'}_1\geq p_1^*$, as she obtains a moral gain and suffers no material losses. That only leaves  $p_2^*<p_1^* \leq \min \{\bar p^{h*}_1,\bar p^{l*}_1\}$. But then, it is Player 2 who can deviate to $p_2'=\bar p^{Q'}_2=p_1^*-\epsilon$, with $\epsilon \rightarrow 0$. Her utility function's first term increases while the second one does not decrease.  Therefore, we cannot have  $p_2^*<p_1^*$.
	
	(2) Profiles where $p_1^*\leq p_2^*$ are not equilibria. Indeed, Player 1 will always be able to deviate from a profile where $p_1^*\leq p_2^*$ and $v_l<p_2^*$ by setting $p_1'=\bar p^{l'}_2=p_2^*-\epsilon$, with $\epsilon \rightarrow 0$. Her utility function's first term increases while the second one does not decrease.  That only leaves $p_1^*\leq p_2^* \leq v_l$. But then, it is Player 2 who can deviate to $p_2'=\bar p^{Q'}_2=p_1^*$. Her utility function's first term increases while the second one does not decrease. Therefore, we cannot have  $p_1^*\leq p_2^*$.
\end{proof}	


\begin{lemma}\label{lemma:incompinfo:valunc:hm:9}
	Consider the bilateral trade game with with heterogeneous buyer valuations and $\kappa \in (0,1)$. Profiles featuring exclusion in both contingent markets are not Nash equilibria in pure strategies.
\end{lemma}	

\begin{proof}	
	Consider profiles where $\bar p^*_{Qi}<p_j^*<\bar p^*_{Rj}; \forall i,j \in \{1,2\}, i\neq j$ and $Q,R \in \{h,l\}, Q\neq R$. By Lemma \ref{lemma:incompinfo:valunc:hm:1}, we can focus only on the subset of profiles where $v_l<p_i^*\leq v_h  \forall i \in \{1,2\}$. Then:
	
	(1) Profiles where prices are different are not equilibria. Suppose, without loss of generality, that $p_1^* < p_2^*$. Then, if $\bar p^*_{h2} < p_2^*$, Player 2 can deviate to $\bar p_{h2}' = p_2^*$ to obtain a moral gain while suffering no material losses. Thus, we can only have $\bar p^*_{h2} \geq p_2^*$. But then, Player 1 finds it profitable to deviate to $p_1' = \bar p_{l1}' = p_2^* - \epsilon$, with $\epsilon < p_2^* - p_1^*$. Her utility function's first term increases while the second one does not decrease. 
	
	(2) But neither are equilibria profiles where prices are the same. Take $\max \{\bar p^*_{l1},\bar p^*_{l2}\}<p_1^*=p_2^*\leq  \min \{\bar p^*_{h1},\bar p^*_{h2}\}$. Then, without loss of generality, Player 1 can deviate to $p_1' = \bar p^{l'}_1 = p_2^*-\epsilon$, with $\epsilon \rightarrow 0$. She thus obtains a strictly positive moral gain that she trades off against a material loss that can be made arbitrarily small. Her utility function's second term increases while the first one does not decrease. 
	
\end{proof}	


\begin{lemma}\label{lemma:incompinfo:valunc:hm:10}
	Consider the bilateral trade game with heterogeneous buyer valuations and $\kappa \in (0,1)$. Profiles featuring pooling in both contingent markets are Nash equilibria only if $\bar p^*_{hi} \in [\left(1-\lambda\right)p^*_{li}+\lambda r,\left(\frac{1}{\lambda}\right)p^*_{li}-\left(\frac{1-\lambda}{\lambda}\right)r] \forall i,j \in \{1,2\}, i\neq j$.
\end{lemma}	

\begin{proof}		
	By Lemmas \ref{lemma:incompinfo:valunc:hm:6} and \ref{lemma:incompinfo:valunc:hm:7}, we can focus on profiles where prices are weakly above cost and below the high valuation. Take a profile with full trade in both contingent markets, so that $\min \{\bar p^*_{hi},\bar p^*_{li}\}<p_j^*$ and $r\leq p_j^*\leq v_h,\forall i,j\in \{1,2\}, i\neq j$.
	
	(1)	If prices are different, the profile is not a Nash equilibrium. To see it, assume without loss of generality that $p_1^* < p_2^*$. Then, if $\bar p^*_{Q2} < p_2^*$ for some $Q \in \{h,l\}$, Player 2 can profitably deviate to $\bar p_{Q2}' \geq p_2^*$. She obtains a moral gain without suffering any material losses. In turn, if $\bar p^*_{Q2} \geq p_2^*$ for all $Q \in \{h,l\}$ are not Nash equilibria either, as Player 1 can profitably deviate to $p_1' = p_2^*$, obtaining a material gain without experiencing moral losses (her utility function's first term increases while the second one does not decrease).   
	
	(2) If prices are equal but above the low valuation, the profile is not a Nash equilibrium. Indeed, suppose that $v_l < p_1^* = p_2^*$. Then, without loss of generality, Player 1 can deviate to $\bar p^{l'}_1 = p_1' = p_2^*-\epsilon$, with $\epsilon \rightarrow 0$. She obtains a strictly positive material gain by not buying from Player 2 in the low-valuation case. This gain outweighs the material loss brought about by the price reduction if $\epsilon$ is small enough. In addition, she suffers no moral losses. Her utility function's first term increases while the second one does not decrease. 
	
	(3) With equal prices, profiles where a player's lowest threshold is above the common price are not Nash equilibria. Indeed, suppose without loss of generality that $p_1^* = p_2^*<\min\{\bar p^*_{l2},\bar p^*_{h2}\}$. Then, Player 1 can deviate to $p_1' = \bar p_{h1}' = \bar p_{l1}'$ Her utility function's first term increases while the second one does not decrease. 	
	
	(4) With equal prices, if $(\bar p^*_{li},\bar p^*_{hi})$ are such that $\bar p^*_{hi} \notin [\left(1-\lambda\right)p^*_{li}+\lambda r,\left(\frac{1}{\lambda}\right)p^*_{li}-\left(\frac{1-\lambda}{\lambda}\right)r]$ for some $i \in \{1,2\}$ then the profile is not an equilibrium. Indeed, suppose without loss of generality that $\bar p^*_{h2} \notin [\left(1-\lambda\right)p^*_{l2}+\lambda r,\left(\frac{1}{\lambda}\right)p^*_{l2}-\left(\frac{1-\lambda}{\lambda}\right)r]$. Then, Player 1 can deviate to $p_1'= \bar p_{h1}' = \bar p_{l1}' = \max \{p^*_{l2},p^*_{l2}\}$. Her utility function's first term increases while the second one does not decrease. 
	
	(5) Profiles where $r\leq p_i^* = p_j^* \leq v_l$ and $\bar p^*_{hi} \in [\left(1-\lambda\right)p^*_{li}+\lambda r,\left(\frac{1}{\lambda}\right)p^*_{li}-\left(\frac{1-\lambda}{\lambda}\right)r] \forall i,j \in \{1,2\}, i\neq j$ are equilibria. Indeed, notice that given $(p_j^*,p^*_{lj},p^*_{hj})$, $(p_i^*,p^*_{li},p^*_{hi})$ maximises \ref{eq:incompinfo:valunc:payoff}. The strategies are then mutual best responses and therefore, these profiles are Nash equilibria.
	
\end{proof}


\begin{proposition}\label{proposition:valunc:hk:maineq}
	The set of Nash equilibria of the bilateral trade game with heterogeneous buyer valuations and $\kappa = 1$ is $p^*_i\leq \min \{\bar p^*_{hi},\bar p^*_{li}\} \forall i \in \{1,2\}$
\end{proposition}			

\begin{proof}
	Given that $r<v_l<v_h$, Expression (\ref{eq:compinfo:hm:ut3}) with $\kappa = 1$ is maximised if and only if $p^*_i\leq \min \{\bar p^*_{hi},\bar p^*_{li}\}$ .
\end{proof}	
\section{Appendix to Section \ref{sec:inef}}

\noindent \textbf{Proof of Proposition \ref{prop:incompinfo:selfish:maineqi}}

The prices set by a given player are completely independent of her threshold and only affect her utility through their interaction with her rival's threshold. Thus, the bargaining game's equilibrium set is all the pairs of profiles that can be formed from the equilibria of the contingent game. 

Consider now the contingent game. The seller's best reply correspondences are:
%
%
%
%
%
%
%
%
%
%
%
%
\begin{equation*}
p^{BR}_{h} = \argmax_{p_{h}} U^{S}(p_{h},p_{l};\bar p) = \begin{cases} \label{eq_seller_h_bri}
p^{BR}_{h}= \overline{p} & \text{if }  \overline{p}_1 \geq r_h \\
p^{BR}_{h}>\overline{p}  & \text{if } \overline{p} <r_h\\
\end{cases}
\end{equation*}

\begin{equation*}
p^{BR}_{l} = \argmax_{p_l} U^{S}(p_{h},p_{l};\bar p) = \begin{cases} \label{eq_seller_l_bri}
p^{BR}_{l}= \overline{p} & \text{if }  \overline{p} \geq r_l \\
p^{BR}_{l}>\overline{p}  & \text{if } \overline{p} <r_l\\
\end{cases}
\end{equation*}

The buyer's best reply correspondence is:	

\begin{equation*}
\begin{cases}
\bar p^{BR} \geq \max\{p_{h},p_{l}\} & \text{if }0\leq p_{l}\leq v_l, 0 \leq p_{h} \leq v_h \\
\bar p^{BR}\geq p_{h} & \text{if }v_l < p_{l} \leq v^e, p_{l}  \leq p_{h} \leq v_h + \dfrac{1-\lambda}{\lambda}(v_l-p_{l})\\
\bar p^{BR}\in [p_{l},p_{h})  & \text{if } 0\leq p_{l}\leq v_l, v_h< p_{h} \\
\bar p^{BR}\in [p_{h},p_{l})  & \text{if } v_l < p_{l}, 0\leq \min\{p_{l},v_h\}\\	
\bar p^{BR} < \min\{p_{h},p_{l}\}& \text{if } v_l < p_{l}, p_{h} > \min\left\{\max\left\{v_h + \dfrac{1-\lambda}{\lambda}(v_l-p_{l}),p_{l}\right\},v_h\right\} \\
\end{cases}
\end{equation*}

If $ \lambda \geq  \lambda_e$, the  best replies intercept whenever: $$\overline{p}^*=p_{h}^*=p_{l}^*\in[r_h,v_e],  \text{ and }$$ $$\overline{p}^*< r_l, p_{l}^* > v_l, p_{h}^*> \min\left\{\max\left\{p_{l}^*,v_h +\dfrac{1-\lambda}{\lambda}(v_l-p_{l}^*)\right\},v_h\right\}.$$		
If $ \lambda <  \lambda_e$, they do so only when: $$\overline{p}^*< r_l, p_{l}^* > v_l, p_{h}^*> \min\left\{\max\left\{p_{l}^*,v_h + \dfrac{1-\lambda}{\lambda}(v_l-p_{l}^*)\right\},v_h\right\}.$$

%

\begin{lemma}\label{lemma:incompinfo:inef:hm:notrade1}
	Consider the quality uncertainty game between two \textit{homo moralis} and assume that $0<v_l<r_l<r_h<v_h$. Then, no-trade profiles where the weakly largest threshold is below own high quality price or above own bad quality price are not Nash equilibria in pure strategies.
\end{lemma}

\begin{proof}
	Assume that $\bar{p}_i<\min\{p_{jh},p_{jl}\} \forall i,j\in\{1,2\}, i\neq j$ and, without loss of generality, that $\bar{p}_2\leq \bar{p}_1$. Then:
	
	(1) Suppose additionally that $\bar{p}_1 < p_{1h}$. If $\bar{p}_2 < \bar{p}_1$, Player 1 can profitably deviate to $p_{1h}'=\bar p_1$. In turn, if $\bar{p}_2=\bar{p}_1$, she can profitably deviate to $\bar p_{1h}'=p'_1\in(\bar p_1,\min\{p_{1l},p_{2h},p_{2l}\})$. In both deviations, the first term of her utility function does not change while the second one strictly increases.
	
	(2) Suppose, alternatively to (1), that $p_{1l}\leq \bar{p}_1$ (given the assumption that $\bar{p}_i<\min\{p_{jh},p_{jl}\} \forall i,j\in\{1,2\}, i\neq j$, this means that $\bar{p}_2 < \bar{p}_1$. Then, for $\bar{p}_2 < \bar{p}_1$, Player 1 can profitably deviate to $p_{1l}' > \bar{p}_1$. The first term of her utility function does not change while the second one strictly increases.
\end{proof}

\begin{lemma}\label{lemma:incompinfo:inef:hm:notrade2}
	Consider the quality uncertainty game between two \textit{homo moralis} and assume that $0<v_l<r_l<r_h<v_h$. Then, profiles where there is no trade in either contingent market are not Nash equilibria in pure strategies.	
\end{lemma}

\begin{proof}
	Assume that $\bar{p}_i<\min\{p_{jh},p_{jl}\}\forall i,j\in\{1,2\}, i\neq j$ and, without loss of generality that $\bar p_1\leq \bar p_2$. First, consider the case where $\bar p_1<\bar p_2$. Then, by Lemma \ref{lemma:incompinfo:inef:hm:notrade1}, only profiles where $p_{2h}\leq \bar p_2<p_{2l}$ can be candidate equilibria. The profiles left when $\bar p_1<\bar p_2$ are then those where $\bar p_1<p_{2h}\leq \bar p_2<\min\{p_{2l},p_{1h},p_{1l}\}$. However, notice that if $p_{2h}\leq v_h$, then Player 1 finds it (weakly) profitable to increase her own threshold to $\bar p'_1=p_{2h}$. After this deviation, the only change is that Player 1 now acquires Player 2's high quality good at price $p_{2h}\leq v_h$, without buying the bad quality item. The first term of her utility function increases while the second one does not change. Thus, no-trade profiles featuring  $\bar p_1<\bar p_2$ and  $\bar p_2\leq v_h$ are not equilibria.
	
	It remains to check whether profiles such that  $\bar p_1<\bar p_2$ and  $p_{2h} > v_h$ are equilibria. The answer is also negative because we have that $p_{2h}\leq \bar p_2 \implies \bar p_2 > v_h>r_h$. With Player 2's threshold above $r_h$, Player 1 can profitably deviate by lowering her own high quality price to $p_{1h}'=\bar p_2$ and sell her good quality item for a profit, without experiencing neither ``moral'' losses nor gains, as both her prices remain above her threshold. The first term of her utility function increases while the second one does not change. I can then affirm that no-trade profiles where players set different thresholds are not equilibria of this game.
	
	Finally, suppose that $\bar p_1=\bar p_2$. Player 1 can profitably deviate to $p_{1h}'=\bar p_1'\in (\bar p_2, \min\{p_{1l},p_{2l},p_{2h}\})$. In this way, she gets to sell the high quality good to ``herself'' without modifying anything else, thus obtaining a ``moral'' gain. The first term of her utility function does not change while the second one strictly increases.
\end{proof}

\begin{lemma}\label{lemma:incompinfo:inef:hm:notrade3}
	Consider the quality uncertainty game between two \textit{homo moralis} and assume that $0<v_l<r_l<r_h<v_h$. Then, profiles with trade of both qualities in one contingent market and no trade in the other one where agents setting different thresholds are not Nash equilibria in pure strategies.	
\end{lemma}

\begin{proof}
	Assume without loss of generality that $\bar{p}^*_1 < \min\{p_{2h}^*,p_{2l}^*\}$, $\bar{p}^*_2\geq \max\{p_{1h}^*,p_{1l}^*\}$ and $\bar{p}^*_1 \neq \bar{p}^*_2$. Then:
	
	(1) Suppose first that $\bar p_2^* < \bar p_1^*$. This implies that $\bar p_2^* <\min\{p_{2h}^*,p_{2l}^*\}$. But then, Player 2 can profitably deviate to $\bar p_1^* < \bar p_2' = p_{2h}' < p_{2l}'$. Her utility function's first term remains unchanged while the second one increases. 	
	
	(2) Assume now that $\bar p_1^* < \bar p_2^*$. Notice first that a profile complying with these assumptions and also with $p_{1l}^* < \bar p_2^*$ cannot be an equilibrium. Indeed, Player 1 can profitably deviate to $p_{1l}' = \bar p_2^*$ and make a material improvement while experiencing no moral losses, or even gains (her utility function's first term increases  while the second one does not decrease). Next, remark that profiles where $p_{2l}^*\leq \bar p_2^*$ or $p_{2h}^*> \bar p_2^*$ are not equilibria. In the former case, Player 2 can $p_{2l}'> \bar p_2^*$ whilst in the latter she can do so to $p_{2h}'\leq  \bar p_2^*$. The deviations entail a moral gain without material losses (her utility function's first term remains unchanged while the second one increases).
	
	We are thus left with profiles where $\bar p_1^* <  p_{2h}^* \leq \bar p_2^* = p_{1l}^* < p_{2l}^*$ and $p_{1h}^* \leq \bar p_2^*$.	No profiles of this kind are equilibria if $p_{1h}^* < \bar p_2^*$. Player 1 can profitably deviate to  $p_{1h}' = \bar p_1' = p_{2h}^* - \epsilon$. She makes a material gain and sustains no moral losses. 
	
	Profiles with $p_{1h}^* = \bar p_2^*$  are not equilibria either. Consider first the case where $p_{1h}^* = \bar p_2^* \leq v_h$. Then we must also have $p_{2h}^* = \bar p_2^*$ since if $p_{2h}^* < \bar p_2^*$, Player 1 can deviate to $\bar p_1' = p_{2h}^* $ and make a material gain. Therefore, the remaining candidates are $p_{1l}^* = p_{1h}^* = p_{2h}^* = \bar p_2^* < p_{2l}^*$. These are not equilibria, as Player 1 can deviate to $ p_{1h}' = \bar p_1' = \bar p_2^* - \epsilon$. She makes a moral gain and sustains a material loss that can be made smaller by choosing the appropriate $\epsilon$. Assume in turn that $p_{1h}^* = \bar p_2^* > v_h$. Then, Player 2 can deviate to $\bar p_2' = p_{2h}' \in (\bar p_1^*,\bar p_2^*)$. She makes a material gain while avoiding moral losses.
\end{proof}



\begin{lemma}\label{lemma:incompinfo:inef:hm:notrade4}
	Consider the quality uncertainty game between two \textit{homo moralis} and assume that $0<v_l<r_l<r_h<v_h$. Then, profiles with trade of both qualities in one contingent market, no-trade in the other one and agents setting the same threshold are not Nash equilibria in pure strategies.	
	%
\end{lemma}

\begin{proof}
	Assume without loss of generality that $ \min\{p_{1h}^*,p_{1l}^*\} \leq \bar{p}^*_2=\bar{p}^*_1<\min\{p_{2h}^*,p_{2l}^*\}$. Then, Player 2 can profitably deviate to $(\bar p_2',p_{2h}',p_{2l}')$ such that $\bar{p}^*_2=\bar{p}^*_1<p_{2h}'\leq \bar p_2'<p_{2l}'$. This simple deviation allows her to achieve a `moral'' gain by commencing to trade with ``herself'' the good quality object, thus breaking the equality in thresholds assumed at first.
	
\end{proof}


\begin{lemma}\label{lemma:incompinfo:inef:hm:bad}
	Consider the quality uncertainty game between two \textit{homo moralis} and assume that $0<v_l<r_l<r_h<v_h$. Then, profiles where only the bad quality is exchanged in at least one contingent market at a price strictly below the threshold are not Nash equilibria in pure strategies.	
\end{lemma}

\begin{proof}
	Assume without loss of generality that  $p_{1l}^*< \bar p^*_2 < p_{1h}^*$. Then, Player 1 can profitably deviate to $p_{1l}' = \bar p^*_2$. She makes a material gain without suffering moral losses.
\end{proof}	

\begin{lemma}\label{lemma:incompinfo:inef:hm:bad1}
	Consider the quality uncertainty game between two \textit{homo moralis} and assume that $0<v_l<r_l<r_h<v_h$. Then, profiles where only the bad quality is exchanged in at least one contingent market at a price above the low quality valuation are not Nash equilibria in pure strategies.
\end{lemma}	

\begin{proof}
	Assume without loss of generality that $v_l < p_{2l}^* \leq  \bar p^*_1 < p_{2h}^*$. By Lemma \ref{lemma:incompinfo:inef:hm:bad}, we can focus on the case $v_l < p_{2l}^* =  \bar p^*_1 < p_{2h}^*$. Then:
	
	(1) Suppose in addition that $p_{1h}^*< \bar p^*_1$. Then Player 1 can profitably deviate to $\bar p_1'\in (p_{1h}^*,\bar p_1^*)$. She makes a material gain while avoiding a moral loss.
	
	(2) Assume, alternatively, that $p_{1h}^*= \bar p^*_1$. Then Player 1 can deviate to $\bar p_1'=\bar p'_1<\bar p_1^*-\epsilon$. The deviation entails a material gain from stopping the trade of the bad quality item at a price above valuation and a potential loss from reducing her own high quality price. However, the loss can be made smaller than the gain by choosing an appropriate $\epsilon$. Since there are no moral losses involved, the deviation is profitable.
	%
	
	(3) Assume, in turn, that  $p_{1h}^* > \bar p^*_1$. Then, Player 1 can profitably deviate to $\bar p'_1 < \bar p^*_1$. She makes a material gain while avoiding a moral loss.
\end{proof}


\begin{lemma}\label{lemma:incompinfo:inef:hm:bad2}
	Consider the quality uncertainty game between two \textit{homo moralis} and assume that $0<v_l<r_l<r_h<v_h$. Then, profiles where only the bad quality is exchanged in at least one contingent market at prices below low quality cost are not Nash equilibria in pure strategies.
	%
\end{lemma}	

\begin{proof}
	Assume without loss of generality that  $p_{2l}^* \leq \bar p^*_1  < \min\{r_l,p_{2h}^*\}$. By Lemma \ref{lemma:incompinfo:inef:hm:bad}, we can focus on the case $p_{2l}^* = \bar p^*_1 < \min\{r_l,p_{2h}^*\}$. Then, Player 2 can profitably deviate to $p_{2l}' > r_l$. She makes a material gain while avoiding a moral loss.
\end{proof}	

\begin{corollary}\label{cor:incompinfo:inef:hm:bad2}
	If we assume $v_l < r_l$, profiles where only the bad quality is exchanged in at least one contingent market are not equilibria of this game. Indeed, Lemmas \ref{lemma:incompinfo:inef:hm:bad}, \ref{lemma:incompinfo:inef:hm:bad1} and \ref{lemma:incompinfo:inef:hm:bad2} show that the threshold of the contingent market featuring only trade of the bad quality cannot be neither above $v_l$ nor below $r_l$ to be equilibria. 
\end{corollary}

\begin{lemma}\label{lemma:incompinfo:inef:hm:bad3}	
	Consider the quality uncertainty game between two \textit{homo moralis} and assume that $0<v_l<r_l<r_h<v_h$. Then, profiles with only good quality trade in one contingent market and no trade in the other one are not Nash equilibria in pure strategies.	
\end{lemma}		
\begin{proof}
	Assume without loss of generality that  $\bar p^*_1 < \min\{p_{2h}^*,p_{2l}^*\}$ and $p_{1h}^*\leq \bar p^*_2 <p_{1l}^*$. 
	
	(1) Consider the case where, in addition, $\bar p^*_2 < \bar p^*_1$. Firstly, Player 1 can profitably deviate from any profile where $\bar p^*_1\geq p_{1l}^*$ to $p_{1l}'>\bar p^*_1$.She makes a moral gain without sustaining material losses. Therefore, I can restrict attention to those where it is also true that $\bar p^*_1<p_{1l}^*$. But then, notice that Player 2 can make a ``moral'' gain by deviating to $p_{2l}' > \bar p'_2 = p_{2h}' \in (\bar p^*_1,p_{1l}^*)$.
	
	(2) In turn, suppose that $\bar p^*_1 < \bar p^*_2$. Firstly, notice that if $p_{2l}^*\leq \bar p^*_2$ Player 2 can profitably deviate to $p_{2l}' > \bar p^*_2$ and make a moral gain without suffering material losses. I then focus on cases where  $p_{2l}^* > \bar p^*_2$. If $\bar p^*_2 < r_h$, Player 1 can profitably deviate to $p_{1h}'\in(\bar p_2^*,\min\{p_{2h}^*,p_{2l}^*\})$. She makes a material gain without sustaining moral losses. If $\bar p^*_2\geq r_h$, Player 1 can profitably deviate to $p_{1l}'=\bar p^*_2$. Again, she makes a material gain without sustaining moral losses. 
	%
	%
	%
	
	(3) Assume now that$\bar p^*_2 = \bar p^*_1$. Then, Player 2 can profitably deviate to $ p_{2l}' > \bar p_2' = p_{2h}' \in (\bar p^*_1,p_{1l}^*)$ and $p_{2h}'=\bar p'_2$. she makes a material gain without sustaining moral losses. 
\end{proof}


\begin{lemma}\label{lemma:incompinfo:inef:hm:bad4}
	Consider the quality uncertainty game between two \textit{homo moralis} and assume that $0<v_l<r_l<r_h<v_h$. Then, only profiles where players set the same high quality price and threshold remain as equilibrium candidates when there is only good quality trade in both contingent markets.
	%
	%
\end{lemma}	

\begin{proof}
	Assume without loss of generality that  $\bar p^*_1 < \bar p^*_2$. Then, Player 2 can always profitably deviate from a profile where $p_{2l}^*  \leq \bar p^*_2$ to $p_{2l}' > \bar p_2^*$, making a moral loss without suffering material losses. Therefore, I consider only profiles where $\bar p^*_2 < p_{2l}^*$. 
	
	Next, notice that Player 1 can profitably deviate from any profile where $p_{1h}^* < \bar p^*_2$ by setting $\bar p_1' = p_{1h}' =\bar p_2^*$. She makes a material gain without experiencing moral losses. As a consequence, I further restrict attention to profiles where $p_{2h}^* \leq \bar p^*_1 < p_{1h}^* =\bar p^*_2 < \min\{p_{2l}^*,p_{1l}^*\}$. Furthermore, Player 2 can profitably deviate whenever $p_{2h}^* < \bar p^*_1 < p_{1h}^* =\bar p^*_2 < \min\{p_{2l}^*,p_{1l}^*\}$ to $p_{2h}' =  \bar p^*_1$, which leaves only $p_{2h}^* = \bar p^*_1 < p_{1h}^* = \bar p^*_2 < \min\{p_{2l}^*,p_{1l}^*\}$. But notice that Player 1 can then deviate to $\bar p'_1 = p_{1h}^* =  \bar p^*_2$ and obtain a material gain without sustaining moral losses.
	
	It follows that out of all possible profiles complying with $p_{ig}^*\leq \bar p^*_j < p_{ib}^* \forall i,j\in\{1,2\},i\neq j$, the only equilibrium candidates are those where  $p_{ig}^* =  \bar p^*_j = p_{jh}^* =\bar p^*_i < \min\{p_{ib}^*,p_{jl}^*\}$.
\end{proof}
\begin{lemma}\label{lemma:incompinfo:inef:hm:bad5}
	Consider the quality uncertainty game between two \textit{homo moralis} and assume that $0 < v_l < r_l < r_h < v_h$. Then, only profiles where players set the same high quality price and threshold and these are between high quality cost and valuation remain as equilibrium candidates when there is only good quality trade in both contingent markets.
	%
\end{lemma}	
\begin{proof}
	Define $p^*: p^* = p_{jh}^* = \bar p^*_i = p_{ig}^* =\bar p^*_j\forall i,j\in\{1,2\},i\neq j$. Suppose first that $v_h < p^* < \min\{p_{1l}^*,p_{2l}^*\}$. Then, without loss of generality, Player 1 can profitably deviate to $\bar p'_1 = p_{1h}' < p^* - \epsilon$. She makes a material loss (she sells the high quality item at a lower price) that can be made smaller than the material gain by choosing an appropriate $\epsilon$. In addition, she does not make any ``moral'' losses.
	
	In turn, assume that $p^* < \min\{r_h,p_{2l}^*,p_{1l}^*\}$. Without loss of generality, Player 1 can profitably deviate by setting $p_1' = p_{1h}'\in (p^*,\min\{p_{2l}^*,p_{1l}^*\})$. She obtains a material gain without sustaining moral losses.
\end{proof}
\begin{lemma}\label{lemma:incompinfo:inef:hm:bad6}
	Consider the quality uncertainty game between two \textit{homo moralis} and assume that $0 < v_l < r_l < r_h < v_h$. Define $p^*: p^* =p_{ig}^* = \bar p^*_j = p_{jh}^* =\bar p^*_i \forall i,j\in\{1,2\},i\neq j$. Then, profiles where $r_h\leq p^* \leq \min\left\{(1-\lambda)r_l + \lambda v_h,\frac{r_l - \kappa v_l}{1-\kappa} \right\}$, $ p^* < \min\{p_{ib}^*,p_{jl}^*\}$ are Nash equilibria in pure strategies if and only if:
	
	$$\lambda v_h + (1-\lambda )r_l\geq r_h \iff \lambda \geq \frac{r_h-r_l}{v_h-r_l}, \text{ and}$$
	
	$$\frac{r_l - \kappa v_l}{1-\kappa} \geq r_h \iff \kappa \geq \frac{r_h-r_l}{r_h-v_l}.$$
\end{lemma}		


\begin{proof}
	Consider a profile where $r_h \leq p^*\leq v_h$ and $ p^* < \min\{p_{1l}^*,p_{2l}^*\}$. Take, without loss of generality, Player 1. She has only two potentially profitable deviations available. On the one hand, she can deviate to $p_{1l}' = p^*$. This allows her to sell the bad quality object to Player 1 at a price above cost $r_l$, but also entails a ``moral'' loss as she would trade the bad quality item with ``herself''. 
	
	On the other hand, she can deviate to $p_{1l}' = p^*$ and $p_{1h}' = \bar p'_1 = p^* - \epsilon$. This allows her to sell the low quality good to Player 2 at a price above cost while not selling it to ``herself'' (avoiding moral losses). In contrast, she stops buying the good quality item from Player 2 at a price below her valuation (because she set a threshold below the latter's price), thus making a material loss. Finally, there is also a material loss that stems from lowering the high quality price but this can be made smaller than the material gain by choosing an appropriate value for $\epsilon$.
	
	\noindent Utility at the initial profile is:
	\begin{equation}\label{eq_ut_pmoral_1g1g_inef0}
	\begin{aligned}
	U^{hm}\big((p^*,p_{1l}^*,p^*),(p^*,p_{2l}^*,p^*)\big)&= \dfrac{1-\kappa}{2}\bigg[\lambda (v_h - p^*)  + \lambda (p^* - r_h) \bigg]+ \dfrac{\kappa}{2}\bigg[\lambda (v_h - r_h)\bigg].
	\end{aligned}
	\end{equation}
	
	\noindent Meanwhile, utilities obtained at respectively the first and second deviations described are:
	\begin{equation}\label{eq_ut_pmoral_1g1g_inef1}
	\begin{aligned}
	U^{hm}\big((p^*,p^*,p^*),(p^*,p_{2l}^*,p^*)\big)&= \dfrac{1-\kappa}{2}\bigg[\lambda (v_h - p^*)  + \lambda( p^*-r_h) + (1-\lambda)(p^*-r_l)\bigg] +\\&
	\dfrac{\kappa}{2}\bigg[\lambda (v_h-r_h) + (1-\lambda) (v_l-r_l)\bigg],
	\end{aligned}
	\end{equation}
	\begin{equation}\label{eq_ut_pmoral_1g1g_inef2}
	\begin{aligned}
	U^{hm}\big((p_{1h}',p^*,\bar p_1'),(p^*,p_{2l}^*,p^*)\big)&= \dfrac{1-\kappa}{2}\bigg[\lambda( p_{1h}'-r_h)+ (1-\lambda) (p^*-r_l)\bigg]+\\&
	\dfrac{\kappa}{2}\bigg[\lambda (v_h - r_h)\bigg].
	\end{aligned}
	\end{equation}
	\noindent Comparing (\ref{eq_ut_pmoral_1g1g_inef0}) with (\ref{eq_ut_pmoral_1g1g_inef1}), we conclude that the first deviation is profitable if and only if:
	\begin{equation}\label{eq_ut_pmoral_1g1g_price1}
	p^*>\frac{r_l - \kappa v_l}{1-\kappa}.
	\end{equation}
	\noindent Meanwhile, the comparison of (\ref{eq_ut_pmoral_1g1g_inef0}) with (\ref{eq_ut_pmoral_1g1g_inef2}) yields that the second one is beneficial if and only if:
	\begin{equation}\label{eq_ut_pmoral_1g1g_price2}
	p^*>r_l + \frac{\lambda}{1-\lambda}(v_h - p_{1h}').
	\end{equation}
	
	\noindent 	By taking into account that $p_{1h}'$ can be made arbitrarily close to $p^*$ we can take the limit of the expression at the right-hand side of (\ref{eq_ut_pmoral_1g1g_price2}) and get:
	\begin{equation}\label{eq_ut_pmoral_1g1g_price2a}
	\lim\limits_{p_{1h}'\rightarrow p^*} r_l + \frac{\lambda}{1-\lambda}(v_h - p_{1h}') =r_l + \frac{\lambda}{1-\lambda}(v_h - p^*).
	\end{equation}
	
	\noindent The second deviation is then profitable if and only if:
	\begin{equation}\label{eq_ut_pmoral_1g1g_price2b}
	p^*>(1-\lambda)r_l + \lambda v_h.
	\end{equation}
	
	\noindent Therefore, the game has equilibria in pure strategies where only the good quality is traded in both contingent markets if and only if:
	\begin{equation}\label{cond1}
	\lambda v_h + (1-\lambda )r_l\geq r_h \iff \lambda \geq \frac{r_h-r_l}{v_h-r_l}, \text{ and}
	\end{equation}
	\begin{equation}\label{cond2}
	\frac{r_l - \kappa v_l}{1-\kappa} \geq r_h \iff \kappa \geq \frac{r_h-r_l}{r_h-v_l},
	\end{equation}
	\noindent They are characterised by:
	$$ r_h  \leq p^*= p_{1h}^* = p_{2h}^* = \bar p^*_2=\bar p^*_1\leq \min\left \{(1-\lambda)r_l + \lambda v_h,\frac{r_l - \kappa v_l}{1-\kappa}\right \} \text{ and }p^*<\min\{p_{2l}^*,p_{1l}^*\}.$$
\end{proof}

%

\begin{lemma}\label{lemma:incompinfo:inef:hm:bad7}
	Consider the quality uncertainty game between two \textit{homo moralis} and assume that $0<v_l<r_l<r_h<v_h$. Then, profiles with trade of the good quality in one contingent market and full trade in the other one are not equilibria if players set different thresholds and prices for the traded qualities. 
	%
	%
\end{lemma}	

%
\begin{proof}

	Assume without loss of generality, that $p_{2h}^*\leq \bar p^*_1 <p_{2l}^*, \max\{p_{1h}^*,p_{1l}^*\} \leq \bar p^*_2$.
	
	(1) Suppose first that in addition to the above, $\bar p^*_2 < \bar p^*_1$. Then Player 2 could always deviate profitably from profiles where $p_{2h}^* < \bar p^*_1$ just by setting $p_{2h}' = \bar p_2' = \bar p^*_1$. This is enough to discard profiles with $\bar p^*_2 <\bar p^*_1$. 
	
	(2) Assume in turn, that $\bar p^*_1 < \bar p^*_2$. No profile where $p_{2l}^* \leq \bar p_2^*$ can be an equilibrium, as Player 2 can deviate to $\bar p^*_2 < p_{2l}'$ and obtain a moral gain without suffering material losses. Therefore, I can circumscribe to profiles where $p_{2l}^* > \bar  p^*_2$. Moreover, $p_{1l}^* < \bar p^*_2$ cannot be part of an equilibrium, as Player 1 can profitably deviate to $p_{1l}' = \bar p^*_2$. She obtains a moral gain without suffering material losses. So we are left with $p_{2h}^* \leq  \bar p^*_1 < \bar p^*_2 = p_{1l}^*  < p_{2l}^*$ and $p_{1h}^*\leq \bar p^*_2$. If $p_{1h}^* < \bar p^*_2$, Player 1 can profitably deviate by setting $p_{1h}' =  \bar p'_1 \in (p_{1h}^*,\bar p^*_2)$. She makes a material gain without suffering moral losses. In turn, if $p_{1h}^* = \bar p^*_2$, she can deviate to $p_{1h}' = \bar p'_1 = \bar p^*_2 - \epsilon$. She experiences a strictly positive moral gain and a material loss. However, the latter can be made smaller than the former by choosing an appropriate $\epsilon$. Then, the thresholds cannot be different.
	
	(3) Finally, suppose that $\bar p^*_1 = \bar p^*_2$ but $p_{1h}^*$, $p_{1l}^*$ or $p_{2h}^*$ are strictly below $\bar p^*_1 = \bar p^*_2$. Then, the agent can just increase her price up to the thresholds and make a material gain while sustaining no moral losses.
\end{proof}

\begin{corollary}\label{cor:incompinfo:inef:hm:badrg}
	A necessary condition for the existence of equilibria with trade of the good quality in one contingent market and full trade in the other one is that $\lambda v_h + (1-\lambda )v_l\geq r_h \iff  \lambda \geq  \lambda_e$.
\end{corollary}

\begin{lemma}\label{lemma:incompinfo:inef:hm:badrg}
	Consider the quality uncertainty game between two \textit{homo moralis} and assume that $0<v_l<r_l<r_h<v_h$. Take all profiles with trade of the good quality in one contingent market and full trade in the other one where agents set the same thresholds and prices for traded qualities. Then, only those with prices above high quality cost and below expected valuation remain as equilibrium candidates.
	%
	%
\end{lemma}

\begin{proof}
	By Lemma \ref{lemma:incompinfo:inef:hm:badrg}, we can focus on profiles where $p_{ig}^* = p_{ib}^* = p_{jh}^* = \bar p^*_i = \bar p^*_j < p_{jl}^*, i,j \in \{1,2\}, i\neq j$. 
	
	Assume without loss of generality, that $p_{1h}^* = p_{1l}^* = p_{2h}^* = \bar p^*_1 = \bar p^*_2 < p_{2l}^*$ and define $p^* = p_{1h}^* = p_{1l}^* = p_{2h}^* = \bar p^*_1 = \bar p^*_2$.
	
	Suppose first that  $p^* < r_h$. Then, without loss of generality, Player 1 can deviate to $p_{1h}' = \bar p'_1 \in (p^*,p_{2l}^*)$. She makes a material gain and avoids moral losses.
	
	In turn, assume that  $p^* > v^e$. Then, Player 2 can deviate to $p_{2h}' = \bar p'_2 = \bar p^*_1 - \epsilon$. She makes a material gain from ceasing to buy both qualities at a price above the expected valuation and a loss from reducing her price, while experiencing no moral losses. Since the loss can be made smaller than the gain by choosing an appropriate $\epsilon$, the deviation is profitable.
\end{proof}


\begin{lemma}\label{lemma:incompinfo:inef:hm:bad8}
	Consider the quality uncertainty game between two \textit{homo moralis} and assume that $0<v_l<r_l<r_h<v_h$. Then, profiles with full trade in one contingent market and trade of only the good quality in the other one if and only if:
	
	$$\lambda \geq \lambda_2 $$
	
	$$\kappa \in [\kappa_1 , \kappa_a(\lambda)] $$
	
	They are characterised by: $$\frac{r_l - \kappa v_l}{1-\kappa}=p_{ig}^*=p_{jh}^*=p_{ib}^*=\bar p^*_j = \bar p^*_i < p_{jl}^* \text{ for } i,j\in\{1,2\},i\neq j.$$ 
\end{lemma}	

\begin{proof}	
	(1) By Lemmas \ref{lemma:incompinfo:inef:hm:bad7} and \ref{lemma:incompinfo:inef:hm:badrg}, the remaining candidate equilibria are such that $r_h \leq p^* = \bar p^*_i = \bar p^*_j = p_{jh}^* = p_{ig}^* = p_{ib}^* \leq v_e, p^*<p_{jl}^*\forall i,j\in\{1,2\},i\neq j$. Corollary \ref{cor:incompinfo:inef:hm:badrg} establishes that $\lambda \geq \lambda_e$ is a necessary condition for the existence of this kind of equilibria.
	
	(2) Assume, without loss of generality, that $r_h\leq p^*=\bar p^*_1 = \bar p^*_2 = p_{2h}^* = p_{1h}^* = p_{1l}^* \leq v_e, p^*<p_{2l}^*$. Player 1's utility at this kind of profiles is specified in Equation (\ref{eq_ut_pmoral1_1g2_inef0}), while that of Player 2 in (\ref{eq_ut_pmoral2_1g2_inef0}).
	
	\begin{equation}\label{eq_ut_pmoral1_1g2_inef0}
	\begin{aligned}
	U^{hm}_1\big((p^*,p^*,p^*),(p^*,p_{2l}^*,p^*)\big)&= \dfrac{1-\kappa}{2}\bigg[\lambda (v_h - p^*) + \lambda (p^* - r_h) + (1-\lambda)(p^* - r_l)\Bigg]+ \\& \dfrac{\kappa}{2}\Bigg[\lambda (v_h-r_h) + (1-\lambda) (v_l-r_l)\bigg],
	\end{aligned}
	\end{equation}
	
	\begin{equation}\label{eq_ut_pmoral2_1g2_inef0}
	\begin{aligned}
	U^{hm}_2\big((p^*,p_{2l}^*,p^*),(p^*,p^*,p^*)\big)&= \dfrac{1-\kappa}{2}\bigg[\lambda (v_h - p^*) + (1-\lambda)(v_l - p^*)+ \lambda (p^* - r_h) + (1-\lambda)r_l \bigg] + \\& \dfrac{\kappa}{2}\bigg[\lambda (v_h - r_h) \bigg].
	\end{aligned}
	\end{equation}
	
	Player 1 may deviate from the profiles described in two potentially beneficial ways: either by marginally reducing her threshold and high quality price or by increasing her low quality price. The first deviation leads Player 1 to stop trading the bad quality item with ``herself'', but also to stop buying the high quality good from Player 2 whilst also selling her own good quality object for a slightly lower price. In turn, the second one leads to her not selling the low quality good to Player 2 but avoiding the exchange of this same item with ``herself''. The utility achieved by Player 1 after the first deviation is described in Expression (\ref{eq_ut_pmoral1_1g2_inef_dev1}), where $p_{1h}'$ is Player 1's high quality price and threshold after deviating. Meanwhile, Player 1's utility at to the second deviation is shown in Equation (\ref{eq_ut_pmoral1_1g2_inef_dev2}), with $p_{1l}'$ representing the low quality price set by Player 1 at the deviation.
	
	\begin{equation}\label{eq_ut_pmoral1_1g2_inef_dev1}
	\begin{aligned}
	U^{hm}_1\big((p_{1h}',p^*,p_{1h}'),(p^*,p_{2l}^*,p^*)\big)&= \dfrac{1-\kappa}{2}\bigg[\lambda (p_{1h}'-r_h) + (1-\lambda)(p^* - r_l)\Bigg]+ \dfrac{\kappa}{2}\Bigg[\lambda (v_h - r_h) \bigg],
	\end{aligned}
	\end{equation}	
	
	\begin{equation}\label{eq_ut_pmoral1_1g2_inef_dev2}
	\begin{aligned}
	U^{hm}_1\big((p^*,p_{1l}',p^*),(p^*,p_{2l}^*,p^*)\big)&= \dfrac{1-\kappa}{2}\bigg[\lambda (v_h - p^*)  + \lambda (p^*-r_h) \Bigg]+ \\ &\dfrac{\kappa}{2}\Bigg[\lambda (v_h - r_h) \bigg].
	\end{aligned}
	\end{equation}
	
	The first deviation is profitable whenever $$(1 - \kappa)\left(\lambda(v_h - r_h)\right) + \kappa(1 - \lambda) (v_l - r_l)> (1 - \kappa)\lambda (p_{1h}' - r_h).$$ 
	
	Recall that the price reduction from $p^*$ to $p_{1h}'$ is a marginal change and Player 1 can make them arbitrarily similar. Taking the limit  as $p_{1h}'$ approaches $p^*$ of the right hand side in the above expression and solving for $p^*$ gives:
	\begin{equation}\label{eq_ut_pmoral1_1g2_inef_price1}
	p^* > v_h + \frac{\kappa (1-\lambda)}{(1-\kappa)\lambda}(v_l-r_l) = A.
	\end{equation}
	\noindent In turn, the second deviation is beneficial to Player 1 whenever: 
	\begin{equation}\label{eq_ut_pmoral1_1g2_inef_price2}
	p^* < \frac{r_l - \kappa v_l}{1-\kappa} = B.	
	\end{equation}
	
	I now turn to deviations available to Player 2. Firstly, she can reduce her low quality price to $p^*$. In this way, she gets to sell her low quality good to Player 1 at price $p^*$ but also exchanges it with ``herself'', making a ``moral'' loss. Secondly, she can, in addition to lowering her low quality good to $p^*$, slightly reduce her threshold and high quality price to $p_{2h}'$. Thus, she gets to sell her low quality good to Player 1 at price $p^*$ and avoid trading it with `herself''. She also avoids buying Player 1's low quality object. In contrast, she stops buying the good quality item from Player 1. Player 2's utility after the first deviation is shown in (\ref{eq_ut_pmoral2_1g2_inef_dev1}), while the second one is displayed in (\ref{eq_ut_pmoral2_1g2_inef_dev2}).
	
	\begin{equation}\label{eq_ut_pmoral2_1g2_inef_dev1}
	\begin{aligned}
	U^{hm}_2\big((p^*,p^*,p^*),(p^*,p^*,p^*)\big)&= \dfrac{1-\kappa}{2}\bigg[\lambda (v_h-p^*) + (1-\lambda)(v_l-p^*) + \lambda (p^*-r_h) + (1-\lambda)(p^*-r_l)\Bigg] \\ & \dfrac{\kappa}{2}\Bigg[\lambda (v_h-r_h) + (1-\lambda) (v_l-r_l)\bigg],
	\end{aligned}
	\end{equation}
	\begin{equation}\label{eq_ut_pmoral2_1g2_inef_dev2}
	\begin{aligned}
	U^{hm}_2\big((p_{2h}',p^*,p_{2h}'),(p^*,p^*,p^*)\big)&= \dfrac{1-\kappa}{2}\bigg[\lambda (p_{2h}' - r_h) + (1-\lambda)(p^*-r_h)\Bigg]+ \dfrac{\kappa}{2}\Bigg[\lambda( v_h -r_h) \bigg].
	\end{aligned}
	\end{equation}
	
	\noindent The first deviation is profitable to Player 2 whenever:
	\begin{equation}\label{eq_ut_pmoral2_1g2_inef_price1}
	p^* > \frac{r_l - \kappa v_l}{1-\kappa} = B.
	\end{equation}		
	\noindent In turn, the second one is beneficial to Player 2 when:
	\begin{equation}\label{eq_ut_pmoral2_1g2_inef_price2}
	p^* > \frac{v^e+(1-\lambda)r_l}{2-\lambda} = C.
	\end{equation}		
	
	\noindent From expressions (\ref{eq_ut_pmoral1_1g2_inef_price2}) and (\ref{eq_ut_pmoral2_1g2_inef_price1}), the only possible value for $p^*$ in this type of equilibria is:
	$$
	p^* = \frac{r_l - \kappa v_l}{1-\kappa}.
	$$
	
	\noindent From (\ref{eq_ut_pmoral1_1g2_inef_price1}) and (\ref{eq_ut_pmoral2_1g2_inef_price2}) we obtain two new conditions for the existence of these equilibria:
	
	$$
	\frac{r_l - \kappa v_l}{1-\kappa} \leq v_h + \frac{\kappa (1-\lambda)}{(1-\kappa)\lambda}(v_l-r_l),
	$$
	
	$$
	\frac{r_l - \kappa v_l}{1-\kappa} \leq \frac{v^e+(1-\lambda)r_l}{2-\lambda}.
	$$
	
	(3)	Putting together all of the above, we have equilibria with full trade in one contingent market and trade of only the good quality in the other one if and only if:
	\begin{equation}\label{condition1}
	B \in \left[r_h, \min \left\{A,C,v^e \right\} \right].
	\end{equation}
	
	(4)	Notice that: $$r_h \leq v^e \implies C \leq v^e.$$
	
	(5)	We can also show that: $$B \leq C \implies C  \leq A.$$
	
	\noindent Firstly, $$B \leq C \iff \kappa \leq \frac{\lambda (v_h -v_l) + v_l - r_l}{\lambda (v_h - r_l) + r_l - v_l} = \kappa_a(\lambda). $$
	
	\noindent Secondly, $$C \leq A \iff \kappa \leq \frac{2 \lambda v_h - \lambda (r_l + v_l)}{2 \left(\lambda \left(v_h - r_l\right) + r_l - v_l \right) =  \kappa_c(\lambda) } .$$
	
	\noindent Finally, $ \kappa_a(\lambda) \leq \kappa_c(\lambda) \iff \lambda \leq 2$, which is true by definition. 
	
	(6) So we can re-write (\ref{condition1}) as:
	\begin{equation}\label{condition2}
	B \in \left[r_h,  C\right] \iff \kappa \in \left[\frac{r_h-r_l}{r_h-v_l}, \kappa_a(\lambda) \right].
	\end{equation}	
	
	(7) The condition for the above interval to be non-empty is: 
	
	\begin{equation}\label{condition3}
	\lambda \geq \frac{r_h-r_l+r_h-v_l}{r_h-r_l+v_h-v_l},
	\end{equation}	
	
	\noindent which is larger than $\frac{r_h-v_l}{v_h-v_l}$ (the value of $\lambda$ necessary for $r_h\leq v^e$). 
\end{proof}

\begin{lemma}\label{lemma:incompinfo:inef:hm:bad9}
	Consider the quality uncertainty game between two \textit{homo moralis} and assume that $0<v_l<r_l<r_h<v_h$. Then, of all profiles with full trade, only those where players set the same thresholds and prices remain as candidate Nash equilibria in pure strategies.
\end{lemma}	

\begin{proof}		
	
	Assume without loss of generality, that $\bar p^*_1 < \bar p^*_2$. Then:
	
	(1) Profiles where $p_{1l}^* < \bar p^*_2$ can immediately be discarded as candidate equilibria, since Player 1 can profitably deviate to $p_{1l}' = \bar p^*_2$ and make a material benefit without suffering moral losses.
	
	(2) Within profiles where $p_{1l}^* = \bar p^*_2$, those where $p_{1h}^* < \bar p^*_2$ are also not equilibria, as Player 1 can profitably deviate to $p_{1h}' = \bar p'_1 \in (p_{1h}^*,\bar p^*_2)$ and make a material benefit without suffering moral losses.
	
	(3) From profiles where $p_{1h}^* = p_{1l}^* = \bar p^*_2$, Player 1 can profitably deviate to $\bar p'_1 = p_{1h}' = \bar p^*_2 - \epsilon$. She makes a moral gain and a material loss. The latter can be made smaller than the latter by choosing an appropriate value for $\epsilon$ and thus, the deviation is profitable.
	
	(4) Finally, profiles where $p_{2Q}^* < \bar p^*_1 = p_{1h}^* = p_{1l}^* = \bar p^*_2$ for some $Q \in \{g,b\}$ are not equilibria, as Player 2 can profitably deviate to $p_{2Q}' = \bar p^*_1$, making a material gain without facing moral losses.
\end{proof}

The proof of Lemma \ref{lemma:incompinfo:inef:hm:bad9} entails an interesting remark. Indeed, profiles where $r_h\leq p_{2l}^* = p_{2h}^* = \bar p^*_1 < p_{1l}^* = p_{1h}^* = \bar p^*_2 \leq v^e$ are equilibria in the bargaining game between two \textit{homo oeconomicus}. It is \textit{homo moralis}' ``impulse'' to set a high quality price she herself would be willing to accept that drives her to deviate, even though the deviation has no material consequences on either player, but only a moral gain for the deviating agent.

I turn now to profiles where both players set the same threshold. Lemma \ref{lemma:incompinfo:inef:hm:bad10} shows that the only equilibria with full trade are those where players set the same prices and thresholds and characterises them.   

\begin{lemma}\label{lemma:incompinfo:inef:hm:bad9prime}
	Consider the quality uncertainty game between two \textit{homo moralis} and assume that $0<v_l<r_l<r_h<v_h$. Then, full trade profiles with symmetric prices and thresholds are not equilibria if they are below the high quality cost or above the expected valuation. 
\end{lemma}	

\begin{proof}
	Define $p^* = \bar p^*_1 = p_{1l}^* = p_{1h}^* = \bar p^*_2 = p_{2l}^* = p_{2h}^*$. 
	
	(1) Assume that $p^* < r_h$. Then, without loss of generality, Player 1 can profitably deviate to $p_{1h}' = \bar p'_1 > p^*$. She makes a material gain without suffering moral losses.
	
	(2) Assume alternatively that $v^e < p^*$. Then, without loss of generality, Player 1 can profitably deviate to $\bar p'_1 = p_{1h}' = p^*-\epsilon$. She makes no moral losses, makes a material gain from not buying both qualities for a price larger than the expected valuation and a loss from reducing her high quality price. However, the loss can be made smaller than the gain by choosing an appropriate value for $\epsilon$. 
\end{proof}

\begin{lemma}\label{lemma:incompinfo:inef:hm:bad10}	
	Consider the quality uncertainty game between two \textit{homo moralis} and assume that $0<v_l<r_l<r_h<v_h$. Then, full trade profiles are Nash equilibria in pure strategies if and only if:
	
	$$ \lambda \geq \frac{r_h- v_l}{v_h- v_l},$$
	
	$$\kappa \leq \min\left\{\kappa_a(\lambda) ,\kappa_b(\lambda)\right\}.$$	
	
	They are characterised by:
	$$ p^*=\bar p^*_i = p_{jh}^* = p_{jl}^* \in \left[\max\left\{r_h, \frac{r_l - \kappa v_l}{1-\kappa}\right\},v^e + \frac{\kappa (1-\lambda)}{1-\kappa}(v_l - r_l)\right], \forall i,j\in\{1,2\},i\neq j.$$

\end{lemma}

\begin{proof}	
	
	%
	
	When considering only profiles where $r_h\leq p^*\leq v^e$, there are no unilateral deviations that could prove beneficial for all values of $\kappa$ and $\lambda$. Consider (without loss of generality) Player 1. The first potentially profitable deviation could be for her to increase her low quality price. In this way, she would stop trading the bad quality good with ``herself'' (thus making a ``moral'' gain), but she would also stop selling it to Player 2, making a material loss. The second and final one is to slightly reduce her threshold and good quality price. This would allow her to stop trading the bad quality item with ``herself'' and reap a ``moral'' benefit. In contrast, she would stop buying from Player 2 at prices below the expected valuation $v^e$ and would also sell her high quality good for an arbitrarily smaller price. Player 1's utility at the candidate profile is:
	
	\begin{equation}\label{eq_ut_pmoral_2_inef0}
	\begin{aligned}
	U^{hm}\big((p^*,p^*,p^*),(p^*,p^*,p^*)\big)&= \dfrac{1-\kappa}{2}\bigg[\lambda (v_h - p^*)  +  (1-\lambda)(v_l - p^*) + \lambda (p^*-r_h) + (1-\lambda)( p^*-r_l) \bigg] +\\ 
	&\dfrac{\kappa}{2}\Bigg[\lambda (v_h-r_h) + (1-\lambda) (v_l-r_l)\bigg].
	\end{aligned}
	\end{equation}
	
	\noindent Meanwhile, her utility at respectably the first and second deviation is shown in expressions (\ref{eq_ut_pmoral_2_inef0_dev1}) and (\ref{eq_ut_pmoral_2_inef0_dev2}):
	
	\begin{equation}\label{eq_ut_pmoral_2_inef0_dev1}
	\begin{aligned}
	U^{hm}\big((p^*,p_{1l}',p^*),(p^*,p^*,p^*)\big)&= \dfrac{1-\kappa}{2}\bigg[\lambda (v_h - p^*)  + (1-\lambda)(v_l - p^*) + \lambda( p^*-r_h )\bigg]+ \\ 
	&\dfrac{\kappa}{2}\bigg[\lambda (v_h-r_h) \bigg]
	\end{aligned}
	\end{equation} 
	\begin{equation}\label{eq_ut_pmoral_2_inef0_dev2}
	\begin{aligned}
	U^{hm}\big((p_{1h}',p^*,\bar p'_1),(p^*,p^*,p^*)\big)&= \dfrac{1-\kappa}{2}\bigg[ \lambda (p_{1h}'-r_h) + (1-\lambda)(p^*-r_l) \bigg]+ \\ 
	&\dfrac{\kappa}{2}\Bigg[\lambda (v_h - r_h)\bigg]
	\end{aligned}
	\end{equation} 
	
	\noindent The first deviation is profitable if and only if:
	\begin{equation}\label{eq_ut_pmoral1_2_inef_price1}
	p^*<\frac{r_l - \kappa v_l}{1-\kappa}= B.
	\end{equation}
	
	\noindent In turn, the second deviation is profitable whenever: 
	$$
	(1-\kappa)\bigg[\lambda (v_h - p^*)  +  (1-\lambda)(v_l - p^*) + \lambda (p^* - r_h)  \bigg] + \kappa(1-\lambda)(v_l - r_l) < (1-\kappa)\lambda(p_{1h}'-r_h) 
	$$
	\noindent Taking the limit  as $p_{1h}'$ approaches $p^*$ of the right hand side in the above expression and solving for $p^*$ gives:
	%
	%
	%
	\begin{equation}\label{eq_ut_pmoral1_2_inef_price2}
	p^* > v^e + \frac{\kappa (1-\lambda)(v_l - r_l)}{1-\kappa}= D.
	\end{equation}
	
	\noindent	Putting together all of the above, we have full-trade equilibria if and only if:
	\begin{equation}\label{condition4}
	p^* \in \left[\max \left\{r_h,B\right\}, \min \left\{v^e,D \right\} \right].
	\end{equation}	
	
	\noindent Delving into (\ref{condition4}), first notice that $D\leq v^e$ for all $\lambda \in [0,1]$ and $\kappa \in (0,1)$. As a consequence, we can write it as:
	\begin{equation}\label{condition5}
	p^* \in \left[\max \left\{r_h,B\right\}, D \right].
	\end{equation}
	
	\noindent As a consequence, a necessary and sufficient condition for these equilibria to exist is that both $B\leq D$ and $r_h \leq D.$  	
	
	\noindent We have that $B\leq D$ if and only if:  
\end{proof}
\begin{equation}\label{condition6}
\kappa \leq \frac{\lambda (v_h - v_l) + v_l - r_l}{\lambda (v_h - r_l) + r_l - v_l} =  \kappa_a(\lambda),
\end{equation}

\noindent while $r_h\leq D$: 
\begin{equation}\label{condition7}
\kappa \leq \frac{\lambda (v_h - v_l)+v_l - r_h}{\lambda (v_h - r_l) + r_l - r_h} = \kappa_b(\lambda).
\end{equation}

The function $\kappa_a(\lambda)$ is continuous and strictly increasing in $\lambda$. Its root is at $\lambda_a = \frac{r_l-v_l}{v_h-v_l} \in (0,1)$. For $\lambda>\lambda_a$, $\kappa_b(\lambda)>0 \iff \lambda>\lambda_e$. Finally, $\lambda_a<\lambda_e$. Therefore, the necessary and sufficient condition to have $ \kappa_a(\lambda)>0$ and $\kappa_b(\lambda)>0$ is $\lambda > \lambda_e$.

Notice that given $\lambda > \lambda_e$, $\kappa_a(\lambda)\geq \kappa_b(\lambda)$ (respectively, $\kappa_a(\lambda) < \kappa_b(\lambda)$) if and only if $\lambda \leq \lambda_2$  (respectively, $\lambda > \lambda_2$).

We can thus conclude with two necessary and sufficient conditions for the existence of equilibria with full trade: $$ \lambda \geq \lambda_e \text{ and }\kappa \leq \min\{\kappa_a(\lambda),\kappa_b(\lambda)\}.$$

\begin{customprop}{4.5*}\label{prop45star}
	Consider the quality uncertainty game between two \textit{homo moralis} and assume that $0<v_l<r_l<r_h<v_h$.
	
	\begin{itemize}
		
		\item Profiles with full trade in one contingent market and trade of only the good quality in the other one are Nash equilibria in pure strategies if and only if:
		
		$$\lambda \geq \lambda_2 $$
		
		$$\kappa \in [\kappa_1 , \kappa_a(\lambda)] $$
		
		They are characterised by: $$\frac{r_l - \kappa v_l}{1-\kappa}=p_{ig}^*=p_{jh}^*=p_{ib}^*=\bar p^*_j = \bar p^*_i < p_{jl}^* \text{ for } i,j\in\{1,2\},i\neq j.$$ 
		
		\item Full trade profiles are Nash equilibria in pure strategies if and only if:
		
		$$ \lambda \geq \frac{r_h- v_l}{v_h- v_l},$$
		
		$$\kappa \leq \min\left\{\kappa_a(\lambda) ,\kappa_b(\lambda)\right\}.$$	
		
		They are characterised by:
		$$ p^*=\bar p^*_i = p_{jh}^* = p_{jl}^* \in \left[\max\left\{r_h, \frac{r_l - \kappa v_l}{1-\kappa}\right\},v^e + \frac{\kappa (1-\lambda)}{1-\kappa}(v_l - r_l)\right], \forall i,j\in\{1,2\},i\neq j.$$ 
		
	\item Profiles where only good quality is traded in both contingent markets are Nash equilibria in pure strategies if and only if:
		\begin{equation}\label{cond1}
		\lambda v_h + (1-\lambda )r_l\geq r_h \iff \lambda \geq \frac{r_h-r_l}{v_h-r_l}, \text{ and}
		\end{equation}
		\begin{equation}\label{cond2}
		\frac{r_l - \kappa v_l}{1-\kappa} \geq r_h \iff \kappa \geq \frac{r_h-r_l}{r_h-v_l},
		\end{equation}
		\noindent They are characterised by:
		$$ r_h  \leq p^*= p_{1h}^* = p_{2h}^* = \bar p^*_2=\bar p^*_1\leq \min\left \{(1-\lambda)r_l + \lambda v_h,\frac{r_l - \kappa v_l}{1-\kappa}\right \} \text{ and }p^*<\min\{p_{2l}^*,p_{1l}^*\}.$$
	\end{itemize}
\end{customprop}
\begin{proposition}\label{prop:incompinfo:inef:hk:maineq}
	Consider the quality uncertainty game between two \textit{homo kantiensis} and assume that $0<v_l<r_l<r_h<v_h$. Then, the set of Nash equilibria in pure strategies is $\left\{(p_{ig},p_{ib},\overline{p}_i)):p_{ig}^* \leq \bar p^*_i < p_{ib}^*\} \forall i \in \{1,2\} \right\}$.
\end{proposition}

\begin{proof}
	Given that $r_h<v_h$ and $r_l>v_l$, the agent's utility function is maximised if and only if $p_{ig}^* \leq \bar p^*_i < p_{ib}^*$.
\end{proof}

\begin{corollary}\label{cor:incompinfo:inef:hk:maineq}
	The equilibria of the game between two \textit{homo kantiensis} always feature trade of the good quality item in at least one contingent market and no trade of the low quality good in the other one.
\end{corollary}

Corollary \ref{cor:incompinfo:inef:hk:maineq} is due to the fact that it will always be true that either $\bar p^*_1\leq \bar p^*_2$ or $\bar p^*_1>\bar p^*_2$. As a consequence, these equilibria always feature trade of the good quality item in at least one of the contingent markets and no trade of the low quality good in the other one. I can thus assert that profiles where no player is trading with the other one are not equilibria in this model, and neither are situations where both Players trade both items with each other (full trade).

\section{Appendix to Section \ref{sec:alt}}

\noindent \textbf{Proof of Proposition \ref{prop:compinfo:alt:maineq}}

 The prices set by a given player are completely independent of her threshold and only affect her utility through their interaction with her rival's threshold. Thus, the bargaining game's equilibrium set is all the pairs of profiles that can be formed from the equilibria of the contingent game. 

Consider now the contingent game. The buyer's and seller's best reply correspondences are:

	\begin{equation*}
	\bar p^{BR} = \argmax_{\bar p_1} U^{B}(\bar p;p) = \begin{cases} 
	\overline{p}\geq p & \text{if }  (v - p) + \alpha(p-r) \geq 0\iff \dfrac{v-\alpha r}{1-\alpha}=v^{alt}\geq p\\
	\overline{p} < p & \text{otherwise}\\
	\end{cases} 
	\end{equation*}
	
	\begin{equation*}
	p^{BR} = \argmax_{p} U^{S}(p;\bar p) = \begin{cases} 
	\overline{p}= p & \text{if }   (v - p) + \alpha(p-r) \geq 0\iff  \dfrac{r-\alpha v}{1-\alpha}=r^{alt}\leq p\\
	\overline{p} < p & \text{otherwise} \\
	\end{cases}
	\end{equation*}
	
	These correspondences intercept whenever $\overline{p}^*= p^*\in[r^{alt},v^{alt}]$ or $\overline{p}^*<r^{alt},p^*>v^{alt}$. When $\alpha \geq \frac{r}{v}$, $r^{alt}\leq 0$ and therefore there does not exist a non-negative threshold that can be lower than $r^{alt}$.
	

	\noindent \textbf{Proof of Proposition \ref{proposition:incompinfo:valunc:alt}}
	
	The prices set by a given player are completely independent of her threshold and only affect her utility through their interaction with her rival's threshold. Thus, the bargaining game's equilibrium set is all the pairs of profiles that can be formed from the equilibria of the contingent game. 
	
	The buyer's best reply correspondences in the contingent game are:

	\begin{equation*}
	\bar p^{h}_{BR}= \argmax_{\bar  p^{h}} U^{B}(\bar p^{l} ,\bar p^{h}; p) = \begin{cases} \label{eq_buyer_h_br_valunc}
	\bar p^{h}_{BR}= p & \text{if }  p \leq v_h^{alt} \\
	\bar p^{h}_{BR}= p & \text{if }  p > v_h^{alt} \\
	\end{cases}
	\end{equation*}
	
	\begin{equation*}
	\bar p^{l}_{BR}= \argmax_{\bar  p^{l}} U^{B}(\bar p^{l} ,\bar p^{h}; p) = \begin{cases} \label{eq_buyer_l_br_valunc}
	\bar p^{l}_{BR}= p & \text{if }  p \leq v_l^{alt} \\
	\bar p^{l}_{BR}= p & \text{if }  p > v_l^{alt} \\
	\end{cases}
	\end{equation*}
	
	The seller's best reply correspondence in the contingent game is:	\begin{equation*}
	\begin{cases}
	p^{BR} > \max\{\bar p^{h},\bar p^{l}\} & \text{ if }\max\{\bar p^{h},\bar p^{l}\} < r_h^{alt} \\
	
	p^{BR} = \bar p^{h} & \text{ if } \bar p^{h}\geq r_h^{alt},\bar p^{l} \leq \min\{\bar p^{h},\lambda \bar p^{h} + (1-\lambda)\}r_l^{alt}\\ 
	
	 & \text{ or } \bar p^{h}\geq r_e^{alt},\bar p^{l} \in [\bar p^{h},\frac{p^{h*}_j-r_h^{alt}}{1-\lambda}]\\ 
	
	p^{BR} = \bar p^{l} & \text{ if } \bar p^{h}\geq r_l^{alt},\bar p^{l} \in [\frac{p^{h*}_j-r_h^{alt}}{1-\lambda},\bar p^{h}]\\ 
	
		 & \text{ or } \bar p^{l} \geq r_l^{alt},\bar p^{h} \leq \left(1-\lambda\right)\bar p^{l}+\frac{\lambda}{\left(1-\alpha\right)}r_h^{alt}\\ 
		 
	\end{cases}
	\end{equation*}
	
	They intersect at the equilibrium profiles described.

\end{appendices}
\end{document}